\documentclass[showpacs,twocolumn,floatfix,superscriptaddress,notitlepage,aps,prx]{revtex4-2}
\usepackage{graphicx}% Include figure files
\usepackage{mathrsfs}
\usepackage{dcolumn}% Align table columns on decimal point
\usepackage{bm}% bold math
\usepackage{suffix}
\usepackage{xurl}
\usepackage{enumitem}
\usepackage{adjustbox} % for vertical alignment
\usepackage{xr}
\usepackage{multirow}
\usepackage{minitoc}
\usepackage{xfrac}
\usepackage{amsmath}
\usepackage{amssymb}
\usepackage{braket}
\usepackage{physics}
\usepackage{amsthm}
\usepackage{natbib}
\usepackage{mathtools}
\usepackage{diagbox}
\usepackage[utf8]{inputenc}
\usepackage[english]{babel}
\usepackage{scalerel}[2014/03/10]
\usepackage{stackengine}
\usepackage{algorithm}
\usepackage{algpseudocode}
\usepackage{tikz}
\usetikzlibrary{quantikz}
\usepackage{caption}
\usepackage{subcaption}

\usepackage{ragged2e}
\usepackage[colorlinks=true,urlcolor=blue,citecolor=blue,linkcolor=blue]{hyperref}% add hypertext capabilities

\newtheorem{assumption}{Assumption}[]

\newtheorem{theorem}{Theorem}[]

\newtheorem{corollary}{Corollary}[]
\newtheorem{lemma}[]{Lemma}

\newtheorem{definition}{Definition}

\theoremstyle{definition}

\makeatletter
\@addtoreset{proofpart}{theorem}
\@addtoreset{proofcase}{theorem}
\makeatother
\renewcommand{\exp}[1]{e^{ #1 }}

\newcommand{\comment}[1]{{}}

\newcommand{\thu}{Department of Mathematics, Tsinghua University,  Beijing 100084, China}
\newcommand{\YMSC}{Yau Mathematical Sciences Center, Tsinghua University,  Beijing 100084, China}
\newcommand{\bimsa}{Yanqi Lake Beijing Institute of Mathematical Sciences and Applications, Beijing 100407, China }
\newcommand{\shanghai}{Shanghai Qi Zhi Institute, Shanghai 200232, China}
\newcommand{\IIIS}{Center for Quantum Information, IIIS, Tsinghua University, Beijing 100084, China}
\newcommand{\Hefei}{Hefei National Laboratory, Hefei 230088, China}
\RenewDocumentCommand{\tr}{m}{\operatorname{Tr} \left( #1 \right)}
\RenewDocumentCommand{\Tr}{m}{\operatorname{Tr} \left( #1 \right)}

\begin{document}

\title{A Unified Frequency Principle for Quantum and Classical Machine Learning}

\author{Rundi Lu}
\thanks{These authors contributed equally to this work.}
\affiliation{\YMSC}
\affiliation{\thu}

\author{Ruiqi Zhang}
\thanks{These authors contributed equally to this work.}
\affiliation{\YMSC}
\affiliation{\thu}
 
\author{Weikang Li}
\affiliation{\IIIS}
% \affiliation{\Hefei}

\author{Zhaohui Wei}
% \thanks{weizhaohui@gmail.com}
\affiliation{\YMSC}
% \affiliation{\thu}
\affiliation{\bimsa}

\author{Dong-Ling Deng}
\thanks{dldeng@tsinghua.edu.cn}
\affiliation{\IIIS}
\affiliation{\shanghai}  
\affiliation{\Hefei}

\author{Zhengwei Liu}
\thanks{liuzhengwei@mail.tsinghua.edu.cn}
\affiliation{\YMSC}
\affiliation{\thu}
\affiliation{\bimsa}

\begin{abstract}

Quantum neural networks constitute a key class of near-term quantum learning models, yet their training dynamics remain not fully understood. Here, we present a unified theoretical framework for the frequency principle (F-principle) that characterizes the training dynamics of both classical and quantum neural networks. Within this framework, we prove that quantum neural networks exhibit a spectral bias toward learning low-frequency components of target functions, mirroring the behavior observed in classical deep networks. 
We further analyze the impact of noise and show that, when single-qubit noise is applied after encoding-layer rotations and modeled as a Pauli channel aligned with the rotation axis, the Fourier component labeled by $\bm{\omega}$ is suppressed by a factor $(1-2\gamma)^{\|\bm{\omega}\|_1}$. This leads to exponential attenuation of high-frequency terms while preserving the learnability of low-frequency structure.
In the same setting, we establish that the resulting noisy circuits admit efficient classical simulation up to average-case error. Numerical experiments corroborate our theoretical predictions: Quantum neural networks primarily learn low-frequency features during early optimization and maintain robustness against dephasing and depolarizing noise acting on the encoding layer. Our results provide a frequency-domain lens that unifies classical and quantum learning dynamics, clarifies the role of noise in shaping trainability, and guides the design of noise-resilient quantum neural networks.

\end{abstract}

\maketitle
% ------------------------ MAIN BODY ------------------------

\section{Introduction}

Deep neural networks (DNNs) are hierarchical computational models whose layered structures learn increasingly abstract representations through gradient-based optimization~\cite{SCHMIDHUBER201585,lecun2015deep,goodfellow2016deep}. Their capacity to extract multi-scale features has driven major advances in image and speech recognition, natural language processing, and scientific computing~\cite{doi:10.1126/science.aag2302,GU2018354,10.1007/s10462-020-09825-6,doi:10.1126/science.aaa8415,10.5555/3295222.3295349,Luo_2019,PhysRevLett.124.020503,10.1145/3422622,hermann2020deep,jiang2022review,jumper2021highly,zhang2025variational}. 
Motivated by these developments and by the pursuit of practical quantum advantages, a wide range of quantum machine-learning approaches have been proposed~\cite{Schuld03042015,biamonte2017quantum,das2019machine,PhysRevResearch.2.033212,lloyd2013quantum,DunjkoPRL,biamonte2017quantum,abbas2021power,liu2021rigorous,PhysRevX.12.021037,gyurik2023exponential,zhao2025entanglement,zhang2024quantumclassical,ye2025quantum}, among which quantum neural networks (QNNs) serve as an important paradigm. QNNs train parameterized quantum circuits on classical or quantum data, enabling expressive quantum models that are compatible with near-term hardware. They have been explored for supervised learning, generative modeling, and quantum many-body applications~\cite{du2025quantum}.

Despite significant methodological advances, a central problem lies in understanding the optimization dynamics that govern how these models learn.
On the classical side, extensive studies have characterized the behavior of DNNs during training, revealing properties of their loss landscapes, uncovering phenomena such as saddle points~\cite{10.5555/2969033.2969154}, the geometry of local minima~\cite{keskar2016large}, and the implicit regularization effects of gradient-based optimizers~\cite{pmlr-v38-choromanska15,JMLR:v19:18-188}. 
These insights have helped explain why large, overparameterized models generalize well despite their expressive power.
On the quantum side, QNNs face optimization challenges that stem from the structure of parameterized quantum circuits themselves. 
Highly expressive circuits can approximate unitary designs, leading to a concentration of measure in high-dimensional Hilbert spaces and consequently to barren plateaus, where gradients vanish exponentially with system size~\cite{mcclean2018barren,cerezo2021cost}. 
This phenomenon further contributes to intricate trade-offs between expressibility and trainability~\cite{PRXQuantum.3.010313,Arrasmith_2022}.
Recent analyses of overparameterized quantum models have provided valuable insights into the training dynamics~\cite{PhysRevLett.130.150601,shirai2024quantum,PRXQuantum.3.030323}, but these remain tied to specific architectures, cost functions, or noise assumptions. This motivates a central question: does a common principle underlie the training dynamics of both classical and quantum neural networks?

In this work, we establish a unified theoretical framework demonstrating that the frequency principle (F-principle) governs early-stage training dynamics in both classical and quantum neural networks.
Although previous empirical and numerical studies have suggested that the F-principle may serve as a common organizing perspective~\cite{luo2019theory,xu2019training,Zhang_2021,xu2024overview}, its status has remained largely phenomenological: gradient-based learners tend to capture the low-frequency components of a target function before fitting its high-frequency structure. This spectral bias was first identified in classical deep networks~\cite{NEURIPS2018_5a4be1fa,NEURIPS2019_0d1a9651,John_Xu_2020,pmlr-v97-rahaman19a,cao2019towards} and has since been linked to their generalization behavior in highly overparameterized regimes~\cite{Zhang_2021}. More recently, numerical evidence indicates that quantum models exhibit analogous low-frequency preference~\cite{xu2024frequency,duffy2025spectral}, suggesting that the F-principle may reflect a deeper, architecture-independent feature of gradient-driven optimization. Here, we rigorously prove that under gradient-flow dynamics, low-frequency modes are learned substantially faster than high-frequency ones, establishing the F-principle as a rigorous and model-agnostic law.
To bridge classical and quantum learning within a single mathematical formulation, we represent classical data over $\mathbb{R}^d$ (or $\mathbb{T}^d$ after encoding) and quantum pure states over the complex projective space $\mathbb{CP}^{d-1}$, showing that both settings share the same fundamental spectral bias within this unified smooth-manifold framework.

\textit{Noise analysis}\textemdash While our first result establishes a unified F-principle for noiseless training dynamics, practical quantum devices inevitably operate in the presence of noise~\cite{wang2021noise,doi:10.1126/sciadv.adr5002,PhysRevLett.133.120603,aharonov2023polynomial,fontana2025classical}. Prior work has shown that specific circuit architectures under dephasing noise exhibit attenuation of high-frequency spectral components~\cite{fontana2025classical}. We further broaden this picture by analyzing a general class of parameterized quantum circuits in which noise is inserted after each rotation gate as a single-term Pauli channel aligned with the gate's rotation axis. In this setting, we prove that such noise universally induces exponential suppression of high-frequency Fourier components, while leaving low-frequency components comparatively intact. As a consequence, the gradients associated with smooth, low-frequency structure remain stable, implying that noisy QNNs continue to learn low-frequency features effectively during the early stages of optimization. We further verify that the time derivative of the low-frequency loss still dominates that of the full loss under this noise model, so the unified F-principle persists in the noisy regime.

\textit{Classical simulation}\textemdash Using the same Pauli path-integral representation, this noise-induced spectral suppression provides a new perspective on the classical simulability of noisy parameterized quantum circuits. Since the contribution of Fourier modes with higher Fourier frequencies decay exponentially under axis-aligned Pauli noise, the effective spectrum of the circuit becomes sharply concentrated in the low-frequency region. This allows us to truncate the Fourier expansion at a logarithmic frequency cutoff, yielding a polynomial-time classical algorithm that approximates noisy circuit expectation values up to a guaranteed average-case error. Accordingly, the frequency-domain viewpoint casts noisy dynamics as an effective band-limiting mechanism, thereby delineating regimes where classical simulations based on low-frequency truncation are suitable and highlighting parameter regimes in which quantum advantage may still persist.

Numerical experiments further corroborate the theoretical picture. In a one-dimensional regression task with a target 
\(f(x)=\frac{1}{M}\sum_{k=1}^{M}[\sin(kx)+\cos(kx)]\) for \(M=40\) on \(x\in[0,2\pi]\), we observe that quantum neural networks trained under rotation-axis-aligned single-term Pauli noise or depolarizing noise consistently learn the low-frequency components of the target function, while high-frequency details are progressively suppressed. This behavior matches our prediction that low-frequency modes dominate the early training dynamics even in the presence of noise. The interplay between the intrinsic F-principle and noise-induced spectral filtering suggests that QNNs retain a notable degree of robustness against encoding-layer noise when the underlying task is governed by smooth structure, potentially reducing the need for expensive error-mitigation procedures in such settings~\cite{Field:87,_1994,olshausen1996emergence,simoncelli2001natural}.

Our paper is organized as follows.  
Section~\ref{sec:2} presents the main results. Section~\ref{sec:2.1} establishes a unified F-principle for both QNNs and DNNs, providing a common framework that connects their early-stage training dynamics.  
Section~\ref{sec:2.3} analyzes the effect of realistic noise, demonstrating the exponential suppression of high-frequency components, its compatibility with the unified F-principle, and the resulting polynomial-time classical simulability of noisy parameterized quantum circuits.  
Section~\ref{sec:QNN_Robustness} reports numerical evidence showing that QNNs remain robust to encoding-layer noise when learning low-frequency structure.  
Finally, Section~\ref{sec:Discussion} summarizes our findings and outlines broader implications and future research directions.

\begin{figure*}[htb!]
 \centering
 \includegraphics[width = \textwidth]{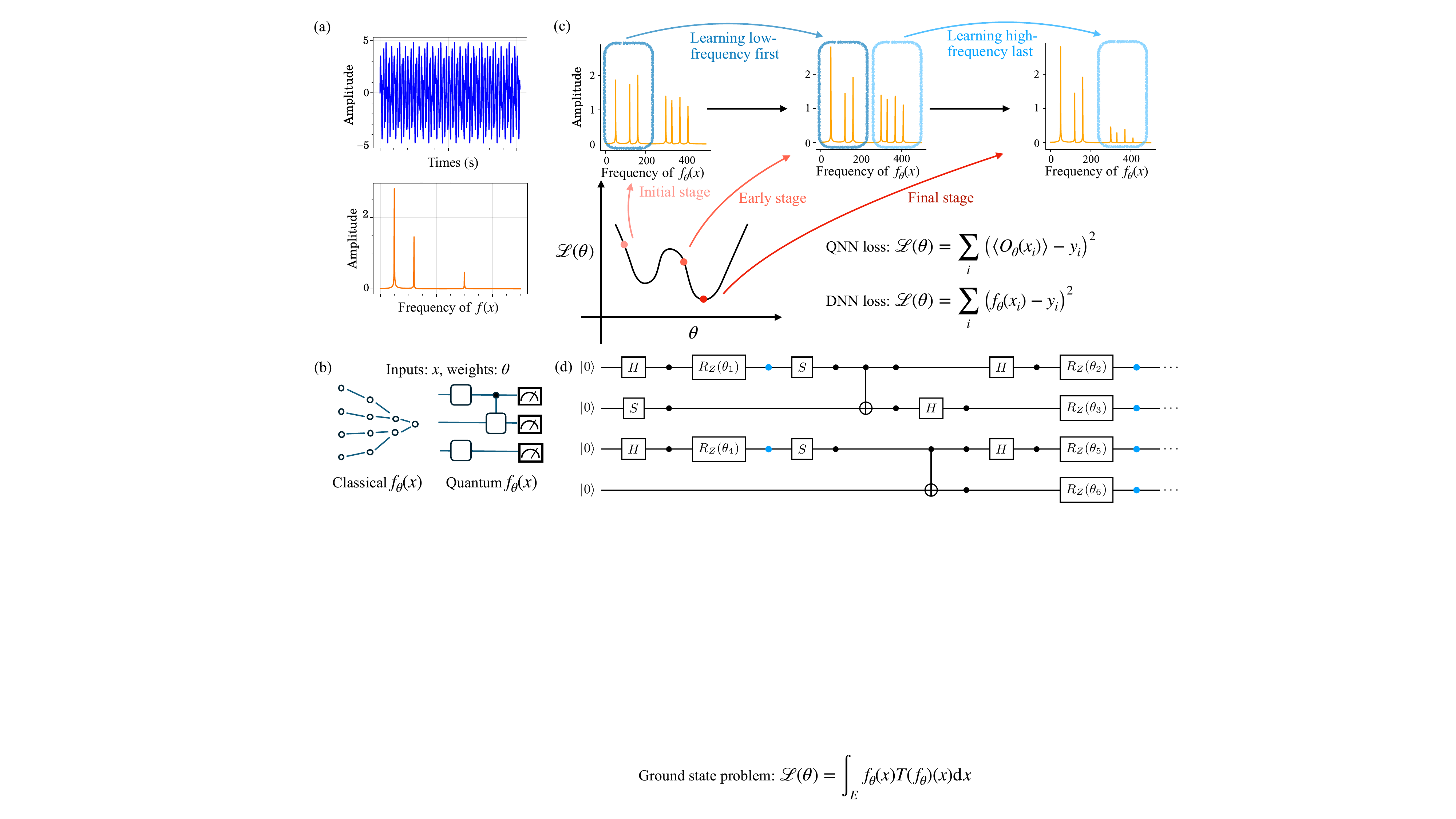}
 \caption{\justifying
Illustration of the F-principle in classical and quantum neural networks.
(a) Upper panel: The time-domain waveform of a target function $f(\bm x)$, exhibiting oscillatory behavior. Lower panel: Its corresponding frequency-domain representation obtained via Fourier transformation, showing distinct frequency components (peaks at different frequencies).
(b) Architectures of classical and quantum neural networks that output $f_{\bm \theta}(\bm x)$, where $\bm x$ denotes the input and $\bm \theta$ denotes the set of tunable parameters. The classical model (left) consists of layered neurons, while the quantum model (right) is implemented as a parameterized quantum circuit followed by projective measurements.
(c) Top: Frequency-domain representations of the QNN output $f_{\bm \theta}(\bm x)$ at different stages of training. From left to right: initial, early, and final stages. The orange lines represent the spectral components of $f_{\bm \theta}(\bm x)$. The dark blue boxes highlight low-frequency components, which are learned rapidly and converge during the early stage. The light blue boxes indicate high-frequency components, whose convergence is delayed until the late stage. This stage-wise convergence behavior illustrates the F-principle.
Bottom: Schematic loss landscapes $\mathcal{L}(\bm \theta)$ for classical and quantum neural network models. The QNN loss is given by $\mathcal{L}(\bm \theta) = \sum_i \left(\langle O_{\bm \theta}(\bm x_i)\rangle - y_i\right)^2$, while the DNN loss takes the form $\mathcal{L}(\bm \theta) = \sum_i \left(f_{\bm \theta}(\bm x_i) - y_i\right)^2$.
(d) Schematic of a parameterized quantum circuit under two noise models. Single-qubit operations ($H$, $S$, $R_Z(\theta_i)$) and CNOT gates are interleaved throughout the circuit. Solid dots mark points at which a local noise channel is applied.  We consider two noise scenarios: (1) dephasing only at the blue-dot positions (black dots remain noise-free), and (2) depolarizing noise at both black and blue dot positions.  The ellipses on the right indicate that the pattern repeats in subsequent layers.
 }
 \label{fig:whole_picture_MPO_QEM}
\end{figure*}

\section{Results}
\label{sec:2}
This section summarizes the main contributions of this work. Section~\ref{sec:2.1} develops a unified F-principle on smooth data manifolds and applies it to both classical deep networks and quantum neural networks, providing a common description of their early training dynamics. Section~\ref{sec:2.3} introduces a frequency-domain analysis of noisy parameterized quantum circuits, establishing the exponential suppression of high-frequency modes, the persistence of the F-principle under realistic noise, and a polynomial-time classical simulation method based on frequency truncation. Section~\ref{sec:QNN_Robustness} reports numerical experiments that demonstrate the predicted low-frequency bias and the noise-induced spectral filtering in a representative regression task.

\subsection{Unified Frequency Principle}
\label{sec:2.1}
This subsection formulates a unified version of the F-principle that applies to both classical and quantum neural networks. We treat QNNs with either classically encoded inputs or genuine quantum input states within the same manifold-based framework: classical data are embedded through rotation gates in the encoding layer, whereas quantum data enter directly as initial states.  The results confirm the validity of the frequency principle across both classical and quantum settings, providing a comprehensive framework for understanding neural network behavior in these domains.

\subsubsection{Frequency Principle}

The F-principle can be formally described by considering a neural network as a parameterized function \( f_{\bm \theta}(\bm x) \), where \( \bm x \) represents the input and \( \bm \theta \) denotes the learnable parameters (weights). During the training process, the network adjusts these parameters to minimize a given loss function, typically the mean squared error:
\begin{equation}\label{eq:continuous_loss}
   \mathcal{L}(\bm \theta) = \int_\mathcal{D}\|f_{\bm \theta}(\bm x)-y(\bm x)\|^2 \rm d\bm x, 
\end{equation}
where $\mathcal{D}$ is the training dataset.

In the frequency domain, we divide the loss function into a low-frequency part $\mathcal{L}_\lambda(\bm \theta) =\int_\mathcal{D}\|E_\lambda(f_{\bm \theta}(\bm x)-y(\bm x))\|^2 \rm dx$ and a high-frequency part $\mathcal{L}(\bm \theta)-\mathcal{L}_\lambda(\bm \theta)$ by applying a spectrum projector $E_\lambda$ that is associated with a specific positive operator onto the frequency
space $[0, \lambda]$. The specific definition of the spectrum projector $E_\lambda$ and the mathematical details are stated in Section~\ref{subsection:: Unified Theorem} and Supp.~\ref{sec:append:new_statement}.
The F-principle asserts that, during training, the neural network initially learns the low-frequency components $\mathcal{L}_\lambda(\bm \theta)$ of the target function, which correspond to its smooth, global features. Subsequently, it progressively adjusts to the high-frequency components, which capture rapid oscillations and finer details. In other words, for smaller values of \( |\lambda| \), $E_\lambda f_{\bm \theta}(\bm x)$ converges to $E_\lambda y(\bm x)$ faster than the larger values of \( |\lambda| \) (see Fig.~\ref{fig:whole_picture_MPO_QEM}).

\subsubsection{Unified Theorem}\label{subsection:: Unified Theorem}
A key technical challenge in developing a unified theorem is that classical and quantum learning models operate on fundamentally different data spaces: standard DNNs process inputs in Euclidean domains (e.g., $\mathbb{R}^d$ or, after periodic encodings, torus $\mathbb{T}^d$), whereas QNNs act on quantum data that naturally live on curved manifolds, such as the complex projective space of pure states $\mathbb{CP}^{d-1}$ (or, more generally, on the manifold of density operators).
To overcome this gap, we abstracted the training problem onto a general smooth manifold that accommodates both cases. By treating the loss function as a function defined on such a manifold and analyzing its evolution under gradient flow, we were able to derive a spectral characterization that applies uniformly across these distinct settings. This geometric perspective provided the conceptual bridge that enabled us to prove a single theorem capturing the frequency principle for both DNNs and QNNs.

We regard the dataset $\{(\boldsymbol{x}_i, y_i)\}_{i=1}^N$ as i.i.d. samples from a smooth manifold $\mathcal{M}$ endowed with a probability measure $\mu$. The target values are generated by a function $y:\mathcal{M}\!\to\!\mathbb{R}$ (or $\mathbb{R}^m$ in the vector-valued case), with $y_i = y(\boldsymbol{x}_i)$. Writing the pointwise error as
$g_{\boldsymbol{\theta}}(\boldsymbol{x}) := f_{\boldsymbol{\theta}}(\boldsymbol{x}) - y(\boldsymbol{x}),$
the empirical squared loss
$\mathcal{L}(\boldsymbol{\theta}) = \frac{1}{N}\sum_{i=1}^N \bigl\|f_{\boldsymbol{\theta}}(\boldsymbol{x}_i) - y_i\bigr\|^2$
is the Monte Carlo estimator of the population risk $
\mathbb{E}_{\boldsymbol{x}\sim \mu}\!\left[\bigl\|g_{\boldsymbol{\theta}}(\boldsymbol{x})\bigr\|^2\right]
= \mathbb{E}_{\boldsymbol{x}\sim \mu}\!\left[\bigl\|f_{\boldsymbol{\theta}}(\boldsymbol{x}) - y(\boldsymbol{x})\bigr\|^2\right].$
Accordingly, the empirical objective $\mathcal{L}(\boldsymbol{\theta})$ is consistent with Eq.~\eqref{eq:continuous_loss} as a Monte Carlo approximation to the continuous loss defined over $\mathcal{M}$.

Training proceeds through methods such as stochastic gradient descent~\cite{pmlr-v28-sutskever13} or Adam~\cite{kingma2014adam}, which can be viewed as discrete gradient-descent iterations. In these methods, the parameters $\boldsymbol{\theta}$ are updated iteratively with learning rate $\eta$. As the step size $\eta \to 0$ and the number of iterations grows, this discrete rule converges to a continuous gradient flow governed by
\begin{equation}\label{eq: train dynamics}
\frac{\rm d\boldsymbol{\theta}(t)}{\rm dt} = -\nabla_{\boldsymbol{\theta}} \mathcal{L}(\boldsymbol{\theta}),
\end{equation}
with initial condition $\boldsymbol{\theta}(0)=\boldsymbol{\theta}_0$. The trajectory $\boldsymbol{\theta}(t)$ thus describes a continuous descent of the loss functional.

At this point, we have established that the training process can be represented as a gradient flow on the data manifold $\mathcal{M}$, where $\mathcal{L}(\boldsymbol{\theta})$ is the population mean-squared error over $\mathcal{M}$. To analyze how errors are reduced across scales in the frequency domain, we introduce a low-high frequency splitting of the error function $g_{\boldsymbol{\theta}}(\boldsymbol{x})=f_{\boldsymbol{\theta}}(\boldsymbol{x})-y(\boldsymbol{x})$ via the spectral resolution of a fixed positive self-adjoint operator on $\mathcal{M}$ (e.g., the Laplace-Beltrami operator $\Delta$; see Supp.~\ref{sec:append:new_statement} for full details). Let $E_\lambda$ denote the spectral projector onto frequencies in $[0,\lambda]$. Operationally, $E_\lambda$ acts on any $h\in L^2(\mathcal{M})$ by retaining only its components with spectral parameter $\xi\le \lambda$, while $(\mathbb{I}-E_\lambda)$ removes these components and retains the complement. We refer to $\lambda>0$ as the spectral cutoff.
With this notation, we define the low-frequency loss as
\[
\mathcal{L}_\lambda(\boldsymbol{\theta}) \;=\; \mathbb{E}_{\boldsymbol{x}\sim \mu}\!\left[\bigl\|E_\lambda g_{\boldsymbol{\theta}}(\boldsymbol{x})\bigr\|^2\right],
\]
and the complementary high-frequency loss as
\[
\mathcal{L}_{>\lambda}(\boldsymbol{\theta}) \;=\; \mathbb{E}_{\boldsymbol{x}\sim \mu}\!\left[\bigl\|(\mathbb{I}-E_\lambda) g_{\boldsymbol{\theta}}(\boldsymbol{x})\bigr\|^2\right].
\]
By orthogonality of the spectral subspaces, the full loss admits the additive decomposition
\[
\mathcal{L}(\boldsymbol{\theta}) \;=\; \mathcal{L}_\lambda(\boldsymbol{\theta}) \;+\; \mathcal{L}_{>\lambda}(\boldsymbol{\theta}).
\]
Consequently, demonstrating that $\mathcal{L}_\lambda(\boldsymbol{\theta}(t))$ decays faster along the gradient flow than $\mathcal{L}_{>\lambda}(\boldsymbol{\theta}(t))$ establishes that low-frequency components of $f_{\boldsymbol{\theta}}$ align with those of $y$ earlier in training-precisely the content of the F-principle. All operator-theoretic definitions and the rigorous spectral identities underlying this decomposition are provided in Supp.~\ref{sec:append:new_statement}, where we also verify that the above splitting is well posed for functions in $H^{m}(\mathcal{M})$ under Assumption~\ref{assumption1}.

In Theorem~\ref{thm:main_theorem1}, we formalize this behavior under mild
regularity. Suppose $\nabla_{\boldsymbol{\theta}}\mathcal{L}(\boldsymbol{\theta}(t))$
is continuous in $t\in[0,T]$ and Lipschitz in $\boldsymbol{\theta}$, and that
there exists positive integer $m$ such that
both the inner product of \( g_{\bm\theta}(\bm x) \) with \( \Delta_{\bm x}^m g_{\bm\theta}(\bm x) \) over \( \mathcal{M} \) and the inner product of \( g_{\bm\theta}(\bm x) \) with \( \Delta_{\bm x}^m \Delta_{\bm\theta} g_{\bm\theta}(\bm x) \)
are uniformly bounded. In addition, we assume the training trajectory satisfies
\(
\bigl|\nabla_{\bm\theta}\mathcal{L}\bigl(\bm\theta(T)\bigr)\bigr|>0
\)
for some \(T>0\). Then there exists $\beta>0$ with
\[
\left|\frac{\rm d\mathcal{L}_\lambda(\boldsymbol{\theta}(t))/\rm dt}{\rm d\mathcal{L}(\boldsymbol{\theta}(t))/\rm dt}\right|
\;\ge\; 1-\beta\,\lambda^{-m},
\]
for all $t\in[0,T]$.
So the decay of the low-frequency loss dominates in early training, providing a
precise manifestation of the F-principle. All differentiability statements hold
almost everywhere with respect to $\mu$, i.e., at every point of the manifold except on a subset of measure zero with respect to the data measure $\mu$.

\begin{theorem}[Informal]\label{thm:main_theorem1}
Let $\mathcal{M}$ be smooth with positive self-adjoint $\Delta_{\boldsymbol{x}}$.
If $g_{\boldsymbol{\theta}}\in H^{m}(\mathcal{M})$ and the above energy bounds
and regularity conditions hold, then for any cutoff $\lambda>0$ the derivative
ratio satisfies
\[
\left|\,1-\frac{\dfrac{\rm d}{\rm dt}\mathcal{L}_{\lambda}(\boldsymbol{\theta}(t))}
                     {\dfrac{\rm d}{\rm dt}\mathcal{L}(\boldsymbol{\theta}(t))}\right|
\le \beta\,\lambda^{-m}.
\]
Consequently, during the early stage of training, the time derivative of the low-frequency loss dominates that of the full loss, formalizing the F-principle. A rigorous statement and proof are given
in Supp.~\ref{sec:append:new_statement}.
\end{theorem}

\subsubsection{Classical Neural Networks}

Having established the unified theorem, we now specialise it to conventional (classical) neural networks.  
Let the input vector \(\bm x \in \mathbb{R}^{d}\) be drawn from a probability distribution \(\mu(\bm x)\) supported on a manifold \(\mathcal{M}\subseteq\mathbb{R}^{d}\).  Common architectures-such as multilayer perceptrons, convolutional neural networks, and transformers-implement a mapping
$f_{\bm\theta}\colon \mathcal{M} \longrightarrow \mathbb{R}^{k}$, 
which is differentiable almost everywhere in both the data \(\bm x\) and the parameters \(\bm\theta\).

Under mild regularity assumptions-namely, that \(\nabla_{\bm\theta}f_{\bm\theta}\) is a.e. Lipschitz and that the spectrum of the data-encoding Jacobian grows at most polynomially-the hypotheses of Theorem~\ref{thm:main_theorem1} are met.  Consequently, the low-frequency component of the training error decays more rapidly than the full error, providing a rigorous confirmation of the F-principle for classical networks.  This theoretical result is in line with, and extends, earlier empirical and analytical studies on the spectral bias of deep learning models~\cite{John_Xu_2020,pmlr-v97-rahaman19a,luo2019theory,cao2019towards,xu2019training,Zhang_2021,xu2024overview}.

\subsubsection{Quantum Neural Networks}

Quantum neural networks provide a natural quantum-mechanical analogue of classical nets and arise most commonly as parameterized quantum circuits~\cite{Benedetti_2019,10.21468/SciPostPhysLectNotes.61}.  
We consider two data modalities and show that each fits into the framework of Theorem~\ref{thm:main_theorem1}.

\paragraph{Classical input data.}
When the inputs are classical, each sample \(\bm x\) is first embedded into the Hilbert space of \(d\) qubits by an encoding circuit \(U_{\text{enc}}(\bm x)\), typically realized as layers of single-qubit rotations interleaved with Clifford gates.  
A trainable ansatz \(U_{\text{ans}}(\bm\theta)\) is then applied, preparing the state
\begin{equation}
        \rho_{\text{out}}(\bm x,\bm\theta)
        = U_{\text{ans}}(\bm\theta)\,
          U_{\text{enc}}(\bm x)\,\rho_0\,
          U_{\text{enc}}(\bm x)^{\dagger}
          U_{\text{ans}}(\bm\theta)^{\dagger},
\end{equation}
where \(\rho_0 = |0\rangle\!\langle 0|^{\otimes d}\) is the computational-basis vacuum.  
The ansatz \(U_{\text{ans}}(\bm\theta)\) has a similar circuit structure to the encoding layer, but with all rotation angles replaced by trainable parameters \(\bm\theta\).
The network output is the expectation value of a Hermitian observable \(O\): $    f_{\bm\theta}(\bm x)=\operatorname{Tr}\!\bigl[\,O\,\rho_{\text{out}}(\bm x,\bm\theta)\bigr].$

Although the raw data lie on the Euclidean manifold \(\mathbb{R}^{d}\), our encoding circuit wraps each real coordinate \(x_{k}\) into a rotation phase
\(
    \varphi_{k}=x_{k}\bmod 2\pi,
\)
or equivalently \(e^{i x_{k}}\).
Because \(\varphi_{k}\) is \(2\pi\)-periodic, every dimension is topologically a circle \(S^{1}\).
The full encoded sample therefore lives on the \(d\)-fold product
\(
    (S^{1})^{d}\equiv \mathbb{T}^{d},
\)
the \(d\)-torus of rotation angles.  
We thus take the data manifold to be \(\mathcal{M}=\mathbb{T}^{d}\).

The smooth dependence of both \(U_{\text{enc}}\) and \(U_{\text{ans}}\) on their arguments, together with the bounded spectrum of \(O\),  fulfills the regularity requirements of Theorem~\ref{thm:main_theorem1}.  
Hence the F-principle holds for QNNs with classically encoded data.  All parameter gradients are evaluated with the parameter-shift rule~\cite{PhysRevA.98.032309,PhysRevA.99.032331,Wierichs2022generalparameter}.

\paragraph{Quantum input data.}
When each training example is already a quantum state \(\rho_{x}\), no encoding layer is required; the parameterized circuit \(U(\bm\theta)\) acts directly on \(\rho_x\): 
    $f_{\bm\theta}(x)=\operatorname{Tr}\!\bigl[\,O\,U(\bm\theta)\,\rho_{x}\,U(\bm\theta)^{\dagger}\bigr]$.
The natural domain manifold is a submanifold of the Hilbert space of density operators, often modelled as the complex projective space \(\mathbb{CP}^{d}\):
For pure‐state inputs \(\rho_{x}=|\psi_{x}\rangle\!\langle\psi_{x}|\) on \(n\) qubits (\(D=2^{n}\) Hilbert dimension), this manifold is naturally identified with the complex projective space
$\mathbb{CP}^{D-1}
      \;:=\;
      \bigl(\mathbb{C}^{D}\!\setminus\!\{0\}\bigr)\big/\!\sim,
      \quad
      \psi \,\sim\, e^{i\phi}\psi$,
the set of one-dimensional complex lines (rays) in \(\mathbb{C}^{D}\) obtained by quotienting out the physically irrelevant global phase.  
For mixed inputs, one instead works on the convex, compact manifold of positive-semidefinite, trace-one Hermitian operators, embedded in the real vector space of Hermitian matrices.

Because the map \((\rho_{x},\bm\theta)\mapsto U(\bm\theta)\rho_{x}U(\bm\theta)^{\dagger}\) is smooth and the observable \(O\) has the bounded spectrum, the regularity conditions of Theorem~\ref{thm:main_theorem1} are satisfied.  Accordingly, the F-principle also governs QNNs trained directly on quantum data.

\subsection{Classical Simulation of Parameterized Quantum Circuits under Generalized Noise} 
\label{sec:2.3}

We next study parameterized quantum circuits with classical inputs under the broad noise model introduced in Eq.~\eqref{eq:generalized_noise_model}. By rewriting noisy expectation values as Pauli path integrals, we obtain a frequency-domain representation in which each Fourier mode is damped by a factor that decays exponentially with its $\ell^1$-frequency, as formalized in Theorem~\ref{thm:frequency_suppression}. This structure underlies both our analysis of the F-principle in the noisy setting (Corollary~\ref{thm: thm1 holds in thm2 case}) and a frequency-truncation scheme that yields an efficient classical-simulation algorithm with a rigorous average-case error guarantee (Theorem~\ref{thm: 3}).

\subsubsection{Noise-Induced Suppression of High-Frequency Components}
In this subsection, we extend the results of Ref.~\cite{fontana2025classical} to a  {broader setting}. 
We show that, under the generalized models stated below, the noise-induced suppression of high-frequency components in parameterized quantum circuits still holds; see Lemma~\ref{lemma:frequency_suppression} in Supp.~\ref{appendix:sec:noise_effect_frequency_components}.

We consider the following circuit model and noise model. 
\paragraph{Circuit model.}
An \(n\)-qubit parameterized quantum circuit of depth \(L\) is written
$C(\bm\theta)=U_L(\theta_L)\,\cdots\,U_1(\theta_1)$,
with parameters \(\bm\theta=(\theta_1,\dots,\theta_L)\).  
Each layer factorizes as
$U_i(\theta_i)=\exp{-i\theta_i P_i/2}C_i$,
where \(C_i\) is a fixed Clifford and $P_i \in \{I, X, Y, Z\}^{\otimes n}$.  
Considering unitary operations as quantum channels, we define 
\(\mathcal{C}(\bm{\theta})(\cdot)=C(\bm{\theta})\,(\cdot)\,C(\bm{\theta})^{\dagger}\).
For observable \(O\) with the bounded spectrum, $ \langle O(\bm{\theta})\rangle = \tr{O C(\bm{\theta})\rho C(\bm{\theta})^\dagger} = \tr{O \mathcal{C}(\bm{\theta})(\rho)},$
where \(\rho\) is the initial state.

\paragraph{Noise model.}
The noise introduced after the rotation gate $\exp{-i \theta_iP_i/2}$ is no longer limited to dephasing noise but takes the more general Pauli channel form
\begin{equation}\label{eq:generalized_noise_model}
    \mathcal{E}_i(\cdot) = (1-\gamma)\,\cdot + \gamma\, P_i(\cdot) P_i^\dagger,
\end{equation}
where $P_i$ is the axis of the rotation gate $U_i(\theta_i)$ and $\gamma$ denotes the noise rate. Here, the noise is applied only after the rotation gates, as shown in Fig.~\ref{fig:whole_picture_MPO_QEM}(d).
Compared with the commonly used noise model of local depolarizing noise with a constant noise rate~\cite{aharonov2023polynomial,PhysRevLett.133.120603,xct1-7kf2}, the noise we consider is weaker but commonly considered~\cite{cao2023generation,van2023probabilistic,kim2023evidence}.
Thus
\(
    \tilde{\mathcal{U}}_i(\theta_i)=\mathcal{E}_i\circ\mathcal{U}_i(\theta_i),
\)
and the noisy circuit is
$\tilde{\mathcal{C}}(\bm{\theta}) =\tilde{\mathcal{U}}_L(\theta_L) \circ \cdots \circ \tilde{\mathcal{U}}_1(\theta_1)  =\mathcal{E}_L \circ \mathcal{U}_L(\theta_L) \circ \cdots \circ \mathcal{E}_1 \circ \mathcal{U}_1(\theta_1)$.
Its expectation value reads $\langle O_{\text{noisy}}(\bm\theta)\rangle=\tr{O \tilde{\mathcal{C}}(\bm{\theta})(\rho) }.$

Rewrite the noiseless expectation value $ \langle O(\bm{\theta})\rangle $ as a Pauli path integral~\cite{aharonov2023polynomial,PhysRevLett.133.120603,lerch2024efficient,angrisani2025simulating,fontana2025classical,rudolph2025pauli}, which takes the form of a trigonometric series in $\bm{\theta}$, we make the spectral information more explicit.
For simplicity, the observable $O$ is taken to be a Pauli operator in $\{I,X,Y,Z\}^{\otimes n}$.  Using the Pauli path integral, one obtains
$  \langle O(\bm{\theta})\rangle 
        =\sum_{s\in \bm{P}^{L+1}_n}\Tr{Os_L}
          \left(\prod_{i=1}^{L}\Tr{s_i\mathcal{U}_i(\theta_i)(s_{i-1})}\right)
          \Tr{s_0 \rho}$, 
  where $s=(s_0,\dots,s_L)$ enumerates Pauli paths in $\bm{P}^{L+1}_n$, $\bm{P}_n = \left\{ \sfrac{\mathbb{I}}{\sqrt{2}}, \sfrac{X}{\sqrt{2}}, \sfrac{Y}{\sqrt{2}}, \sfrac{Z}{\sqrt{2}} \right\}^{\otimes n}$ and $n$ is qubit number.
Define $f(s,\bm{\theta})=
    \Tr{Os_L}
    \Bigl(\prod_{i=1}^{L}\Tr{s_i\mathcal{U}_i(\theta_i)(s_{i-1})}\Bigr)$.
Each term $\Tr{s_i\mathcal{U}_i(\theta_i)(s_{i-1})}$ is either a constant or a constant multiplied by $\sin(\theta_i)$ or $\cos(\theta_i)$ (see Supp.~\ref{app:sec:impact_of_noise_on_frequency_spectrum}). Hence, $f(s,\bm{\theta})$ is a trigonometric polynomial in the parameters $\bm{\theta}$. The number of sine and cosine factors is the frequency of the $f(s,\bm{\theta})$, denoted as $\#\mathrm{trig}(s)$. For different paths $s$, the trigonometric polynomials $f(s,\bm{\theta})$ are orthogonal to each other with respect to the inner product defined in $L^2(\mathbb{T}^L)$, where $\mathbb{T}^L$ is the $L$-dimensional torus.

In the presence of the noise model defined in Eq.~\eqref{eq:generalized_noise_model}, the expectation value of the observable becomes
\begin{equation*}
    \begin{aligned}
        \langle O_{\text{noisy}}(\bm\theta) \rangle
          & = \Tr{O\mathcal{E}_L\circ \mathcal{U}_L(\theta_L) \circ \cdots \circ
                  \mathcal{E}_1 \circ \mathcal{U}_1(\theta_1)(\rho)} \\
          &=\sum_{s\in \bm{P}^{L+1}_n}\tilde f(s,\bm{\theta})\Tr{s_0 \rho},
    \end{aligned}
\end{equation*}
where $  \tilde f(s,\bm{\theta})
     = \Tr{Os_L}
              \Bigl(\prod_{i=1}^{L}\Tr{\mathcal{E}_i(s_i)
                    \mathcal{U}_i(\theta_i)(s_{i-1})}\Bigr)$. Comparing $\tilde f(s,\bm{\theta})$ with $f(s,\bm{\theta})$ reveals the suppression mechanism.  For any non-zero trace term $\Tr{\mathcal{E}_i(s_i)\mathcal{U}_i(\theta_i)(s_{i-1})}$, we have:
\begin{enumerate}%[label=(A\arabic*)]

\item It contains a $\sin(\theta_i)$ or $\cos(\theta_i)$ factor iff $P_i$ anti-commutes with $s_{i-1}$.
\item In that case, the trace is multiplied by $(1-2\gamma)$; otherwise, it is unaffected.
\end{enumerate}
Hence
$\tilde f(s,\bm{\theta})
      =(1-2\gamma)^{\#\mathrm{trig}(s)}\,f(s,\bm{\theta})$,
where $\#\mathrm{trig}(s)$ counts the number of trigonometric factors in $f(s,\bm{\theta})$, which is the frequency of $f(s,\bm{\theta})$.

In a QNN with classical inputs, the parameters naturally split into the input vector $\bm{x}$ and the trainable circuit parameters $\bm{\theta}$. 
Fixing $\bm{\theta}$ and treating $\bm{x}$ as variables, we define
$f_{\bm{\theta}}(\bm{x}) := \langle O(\bm{x};\bm{\theta})\rangle$ as
the expectation value of the parameterized quantum circuit. 
Its noisy counterpart is $f_{\mathrm{noisy},\bm{\theta}}(\bm{x}) := \langle O_{\mathrm{noisy}}(\bm{x};\bm{\theta})\rangle.$ We assume the dimension of $\bm x$ is $l$, and the dimension of parameter $\bm \theta$ is then $L-l$.
Under this viewpoint, and building on the preceding derivation, we obtain the following theorem, which generalizes the result in Ref.~\cite{fontana2025classical}. 
The complete proof is given in Supp.~\ref{app:sec:impact_of_noise_on_frequency_spectrum}.

\begin{theorem}\label{thm:frequency_suppression}
Consider the circuit and observable $O$ defined above. 
If the noise introduced after each rotation gate $\exp{-i P_i\theta_i/2}$ is modeled by a Pauli channel $\mathcal{E}_i$ defined in Eq.~\eqref{eq:generalized_noise_model} with noise rate $\gamma$, then the noisy expectation value decomposed in the frequency domain is given by 
\begin{equation}\label{Eq: thm2 main eq}
    f_{\mathrm{noisy}, \bm{\theta}}(\bm{x}) = \sum_{\bm{\omega} \in \mathbb{Z}^l} (1 - 2\gamma)^{\|\bm{\omega}\|_1} \hat{f}_{\bm{\theta}}(\bm{\omega})\, e^{i \bm{\omega} \cdot \bm{x}},
\end{equation}
where $\hat{f}(\bm{\omega})$ are the Fourier coefficients of the noiseless expectation value
\begin{equation}
    f_{\bm{\theta}}(\bm{x}) = \sum_{\bm{\omega} \in \mathbb{Z}^l} \hat{f_{\bm{\theta}}}(\bm{\omega}) \, e^{i \bm{\omega} \cdot \bm{x}},
\end{equation}
and $\|\bm{\omega}\|_1$ denotes the $\ell_1$-norm of $\bm{\omega}$. 
\end{theorem}
Formally, since the parameter space is a torus $\mathbb{T}^l$, the natural Fourier index set is $\mathbb{Z}^l$. However, for our Pauli-aligned ansatz, the Fourier support collapses to a finite subset, in fact $\{0,1\}^l$, because each rotation contributes only $0/1$ frequency along each parameter.

Theorem~\ref{thm:frequency_suppression} shows that, under the noise model in Eq.~\eqref{eq:generalized_noise_model}, the amplitude of each frequency component is exponentially suppressed by $(1-2\gamma)^{\|\boldsymbol{\omega}\|_1}$. As a result, high-frequency components decay rapidly, while low-frequency components remain comparatively stable. This provides a complementary perspective to Theorem~\ref{thm:main_theorem1}: in this setting, it is the low-frequency part of the noisy expectation value itself, rather than its time derivative, that dominates the overall behavior.

To demonstrate that Theorem~\ref{thm:main_theorem1} also extends to the noisy case, we further show that the time derivative of the low-frequency component,
$\mathrm{d}f_{\mathrm{noisy}, \boldsymbol{\theta}(t), \lambda}(\boldsymbol{x})/\mathrm{d}t$,
dominates that of the full noisy expectation value,
$\mathrm{d}f_{\mathrm{noisy}, \boldsymbol{\theta}(t)}(\boldsymbol{x})/\mathrm{d}t$.
Specifically, we verify that the key conditions (Eqs.~\ref{eq: C_1}-\ref{eq: C_2})  in Theorem~\ref{thm:main_theorem1} remain valid for the noisy function $f_{\mathrm{noisy}, \boldsymbol{\theta}}$. 
We define a positive self-adjoint operator $\Delta_{\boldsymbol{x}}$ by its action on the Fourier basis
$\Delta_{\boldsymbol{x}}\!\left[\hat{f}_{\boldsymbol{\theta}}(\boldsymbol{\omega}) e^{i\boldsymbol{\omega}\cdot \boldsymbol{x}}\right]
    := (1 - 2\gamma)^{-\|\boldsymbol{\omega}\|_1/m}\,
       \hat{f}_{\boldsymbol{\theta}}(\boldsymbol{\omega}) e^{i\boldsymbol{\omega}\cdot \boldsymbol{x}},$
so that $\{\hat{f}_{\boldsymbol{\theta}}(\boldsymbol{\omega}) e^{i\boldsymbol{\omega}\cdot \boldsymbol{x}}\}_{\boldsymbol{\omega}}$ forms an eigenbasis with eigenvalues
$\lambda_{\boldsymbol{\omega}} = (1 - 2\gamma)^{-\|\boldsymbol{\omega}\|_1/m}$.
It follows that
$\bigl|\langle f_{\mathrm{noisy}, \boldsymbol{\theta}}, \Delta_{\boldsymbol{x}}^{m}
f_{\mathrm{noisy}, \boldsymbol{\theta}}\rangle_{L^{2}(\mathcal{M})}\bigr|
  = \Bigl|\sum_{\boldsymbol{\omega}\in\mathbb{Z}_2^l}
           \lambda_{\boldsymbol{\omega}}^{m}(1-2\gamma)^{\|\boldsymbol{\omega}\|_1}\Bigr|
  \le 2^{l},$
and $
\bigl|\langle f_{\mathrm{noisy}, \boldsymbol{\theta}},
\Delta_{\boldsymbol{x}}^{m}\Delta_{\boldsymbol{\theta}}
f_{\mathrm{noisy}, \boldsymbol{\theta}}\rangle_{L^{2}(\mathcal{M})}\bigr|
  = \Bigl|\sum_{\boldsymbol{\omega}\in\mathbb{Z}_2^l}
           \|\boldsymbol{\omega}\|_2^{2}\lambda_{\boldsymbol{\omega}}^{m}(1-2\gamma)^{\|\boldsymbol{\omega}\|_1}\Bigr|
  \le l^{2} 2^{l}.$
Hence, setting $C_{1}=2^{l}$ and $C_{2}=l^{2}2^{l}$ ensures that the assumption of Theorem~\ref{thm:main_theorem1} holds in the noisy case. This leads to the following corollary.

\begin{corollary}\label{thm: thm1 holds in thm2 case}
Consider the circuit, observable $O$, and noise model defined above.  
If the training dynamics of $\boldsymbol{\theta}$ satisfy the assumptions of
Theorem~\ref{thm:main_theorem1}, then there exists a constant $\beta>0$ such that
\[
\left|\frac{\mathrm{d}f_{\mathrm{noisy}, \boldsymbol{\theta}(t), \lambda}(\boldsymbol{x})/\mathrm{d}t}
            {\mathrm{d}f_{\mathrm{noisy}, \boldsymbol{\theta}(t)}(\boldsymbol{x})/\mathrm{d}t}\right|
\ge 1 - \beta\,\lambda^{-m}.
\]
\end{corollary}

Corollary~\ref{thm: thm1 holds in thm2 case} indicates that the gradient of the noisy loss remains dominated by its low-frequency components, consistent with the F-principle even in the presence of these noise channels. Together, Theorem~\ref{thm:frequency_suppression} and Corollary~\ref{thm: thm1 holds in thm2 case} naturally lead to the following question: during the early training stage, if the noise in Eq.~\eqref{eq:generalized_noise_model} is retained only in the encoding layer while mitigated elsewhere, to what extent does the QNN’s performance deteriorate?

We address this question through numerical simulations. The results reveal that, in the frequency domain, QNNs retain strong robustness to encoding-layer noise when learning low-frequency features. Consequently, at early training stages, QNNs subject to encoding-gate noise modeled by Eq.~\eqref{eq:generalized_noise_model} achieve performance comparable to their noiseless counterparts.

\subsubsection{Efficient Classical Simulation via Frequency Truncation}

By extending these techniques and applying Theorem~\ref{thm:frequency_suppression}, we will demonstrate that if the PQC contains the noise model specified in Eq.~\eqref{eq:generalized_noise_model} (where we also allow additional Pauli noise within the PQC), then the circuit can be classically simulated efficiently on average. 
Note that the noise model we consider is weaker than the local depolarizing noise model, which has been widely adopted in prior works~\cite{aharonov2023polynomial,PhysRevLett.133.120603,xct1-7kf2}.
The noisy expectation value $\langle O_{\text{noisy}}(\bm\theta,\bm x)\rangle$ can be expressed as the Pauli path integral  
$\langle O_{\text{noisy}}(\bm\theta,\bm x)\rangle
      = \sum_{s\in \bm{P}^{L+1}_n}
        (1-2\gamma)^{\#\mathrm{trig}(s)}\,
        f(s,\bm{\theta})\Tr{s_0\rho }$,
where \(\#\mathrm{trig}(s)\) counts the trigonometric factors in \(f(s,\bm{\theta})\) and \(\gamma\) is the single-qubit dephasing rate.  
Because the damping factor \((1-2\gamma)^{\#\mathrm{trig}(s)}\) decays exponentially with \(\#\mathrm{trig}(s)\), high-frequency contributions contribute negligibly.  
We therefore truncate the sum to paths with at most \(\eta\) trigonometric terms, defining the low-frequency approximation
\begin{equation}\label{eq:low_frequency_approximation}
    \langle O_{\text{noisy}}\rangle^{(\eta)}
      = \!\!\!\!
        \sum_{s\in \bm{P}^{L+1}_n,\;\#\mathrm{trig}(s)\le\eta}
        \!\!\!\!
        (1-2\gamma)^{\#\mathrm{trig}(s)}
        f(s,\bm{\theta})\Tr{s_0\rho },
\end{equation}
where \(\eta\) is the frequency threshold.  

Using Pauli propagation~\cite{PhysRevLett.133.120603,lerch2024efficient,angrisani2025simulating,fontana2025classical,rudolph2025pauli}, a technique to efficiently enumerate non-zero Pauli path from low frequency to high frequency, it follows that the number of non-zero terms with \(\#\mathrm{trig}(s)\le\eta\) is \(\mathcal{O}(2^{\eta})\); evaluating each coefficient also costs \(\mathcal{O}(2^{\eta})\) operations (we show those details in Supp.~\ref{sec:theoretical_foundations_low_frequency_simulation_noisy_quantum_circuits}).

To quantify the truncation error, we bound the mean-squared deviation: $\mathbb{E}_{\bm\theta}\!\left[ \left| \langle O_{\mathrm{noisy}}\rangle - \langle O_{\mathrm{noisy}}\rangle^{(\eta)} \right|^2 \right]
    \;\le\; (1-2\gamma)^{2\eta}\,\left\|  O \right\|^{2}_2$,
where $\left\|  O \right\|_2$ is the operator norm of $O$ and the inequality uses the orthogonality of distinct Pauli paths,
$\mathbb{E}_{\bm\theta}\!\big[f(s,\bm\theta)\,f(s',\bm\theta)\big]=0$ for $s\neq s'$.
Assuming \(\|O\|_2\le C\), selecting \(\eta=\mathcal{O}\!\bigl(\ln(\varepsilon)\big/\!\ln(1-2\gamma)\bigr)\)  
guarantees  
$ \mathbb{E}_{\bm\theta}
  \bigl|
    \langle O_{\text{noisy}}\rangle
    -\langle O_{\text{noisy}}\rangle^{(\eta)}
  \bigr|^{2}
  \le\varepsilon$.
Note that for each Pauli path \(s\), the computational complexity of evaluating its coefficient \(f(s, \bm{\theta})\) is polynomial in the number of qubits \(n\) and the circuit depth \(L\).  
Therefore, assuming a constant noise rate \(\gamma\) with noise model defined in Eq.~\eqref{eq:generalized_noise_model}, the cost of computing Eq.~\eqref{eq:low_frequency_approximation} is \(\mathcal{O}(2^{\eta}) \cdot \mathrm{poly}(n, L)\),  
which implies an overall runtime that is polynomial in the circuit depth \(L\), the number of qubits \(n\), and \(1/\varepsilon\).

\begin{theorem}\label{thm: 3}
Consider a parameterized quantum circuit and noise model in Theorem~\ref{thm:frequency_suppression}. For  
\(\eta=\mathcal{O}\!\bigl(\ln(\varepsilon)\big/\!\ln(1-2\gamma)\bigr)\)  
the truncated expectation value \(\langle O_{\text{noisy}}\rangle^{(\eta)}\) satisfies  
\[
  \mathbb{E}_{\bm\theta}
  \bigl|
    \langle O_{\text{noisy}}\rangle
    -\langle O_{\text{noisy}}\rangle^{(\eta)}
  \bigr|^{2}
  \le \varepsilon .
\]
The computational cost of evaluating \(\langle O_{\text{noisy}}\rangle^{(\eta)}\) is \(\mathcal{O}(2^{\eta}) \cdot \mathrm{poly}(n, L)\). Thus, for constant \(\gamma\), the noisy circuit expectation can be efficiently simulated.
\end{theorem}

\subsection{Numerical Results}\label{sec:QNN_Robustness}
We complement the analytical results with numerical experiments that probe how noise affects the frequency-domain learning behavior of QNNs. Throughout this section we use the circuit and noise models of Section~\ref{sec:2.3}, specialized to a one-dimensional regression task with classical inputs.
The QNN output \( f_{\bm{\theta}}(\bm{x}) \) is the expectation value of a parameterized quantum circuit, with input-encoding rotations depending on \( \bm{x} \) and trainable rotations parameterized by \( \bm{\theta} \) in the ansatz layer.

\subsubsection{Training Task and Numerical Set-up}\label{sec:numerical_setting}

To illustrate the frequency-domain behavior of noisy QNNs, we consider the single-variable regression task
$ f(x)=\frac{1}{M}\sum_{k=1}^{M}\!\bigl[\sin(kx)+\cos(kx)\bigr]$, where $M=40$, defined on the interval \([0,2\pi]\).

\paragraph*{Circuit architecture.}
The QNN is implemented by a parameterized quantum circuit composed of alternating \(R_Y\) and \(\mathrm{CZ}\) layers (Fig.~\ref{fig:pic_QNN_Robustness_figa}).  
Data are encoded through \(R_Y(x)\) rotations, while trainable \(R_Y(\theta_i)\) rotations form the ansatz.  
Following the block-encoding strategy of Refs.~\cite{perez2020data,10.21468/SciPostPhysLectNotes.61,10.5555/3737916.3741059}, the data rotations are repeated to enhance non-linearity~\cite{PhysRevA.103.032430}.  
All \(R_Y\) gates are compiled into the native gate set \(\{S,H,S^{\dagger},R_Z\}\) to match superconducting hardware constraints.  
The measured observable is \(Z_1=Z\otimes I^{\otimes(n-1)}\); an overall scale factor \(a\) is introduced to match the target amplitude:
$ f_{\bm{\theta},a}(x)=a\,\bigl\langle Z_1(x,\bm{\theta})\bigr\rangle$.

\begin{figure}[htbp!]
 \centering
 \includegraphics[width = \columnwidth]{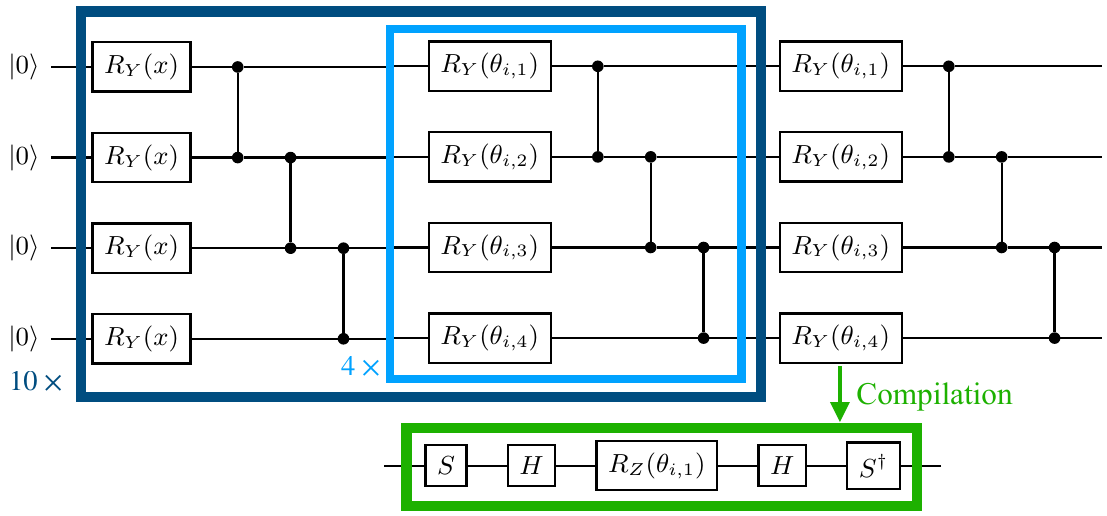}
 \caption{\justifying
The QNN ansatz consists of two types of rotation gates: $R_Y(x)$ gates encode the input data, while $R_Y(\theta_i)$ gates contain trainable weights. To enhance the nonlinear expressive power of the QNN, we apply repeated encoding of the input $x$: the circuit module within the dark blue box is repeated 10 times. Additionally, to increase the number of trainable parameters and thereby improve the model's expressive capacity, the module within the light blue box is repeated 4 times. It is important to note that all rotation gates involving trainable parameters are independently parameterized, whereas all gates involving the input $x$ share the same parameter value across repetitions.
} 
 \label{fig:pic_QNN_Robustness_figa}
\end{figure}

\paragraph*{Data set and optimization.}
A uniformly spaced training set of \(N=2048\) samples \(x_j\in[0,2\pi]\) with targets \(y_j=f(x_j)\) is employed.  
The parameters \((\bm{\theta},a)\) are learned by minimizing the mean-squared error  
$\mathcal{L}(\bm{\theta},a)=\frac{1}{N}\sum_{j=1}^{N}
        \bigl[f_{\bm{\theta},a}(x_j)-y_j\bigr]^2$,
using the Adam optimizer~\cite{kingma2014adam} with a learning rate of \(0.1\) and full-batch updates, which is feasible given the small number of trainable parameters.  
Optimization is halted once \(\mathcal{L}\) plateaus, typically after \(5\times10^{2}\) iterations.

\begin{figure*}[htpb!]
 \centering
 \includegraphics[width = 1\textwidth]{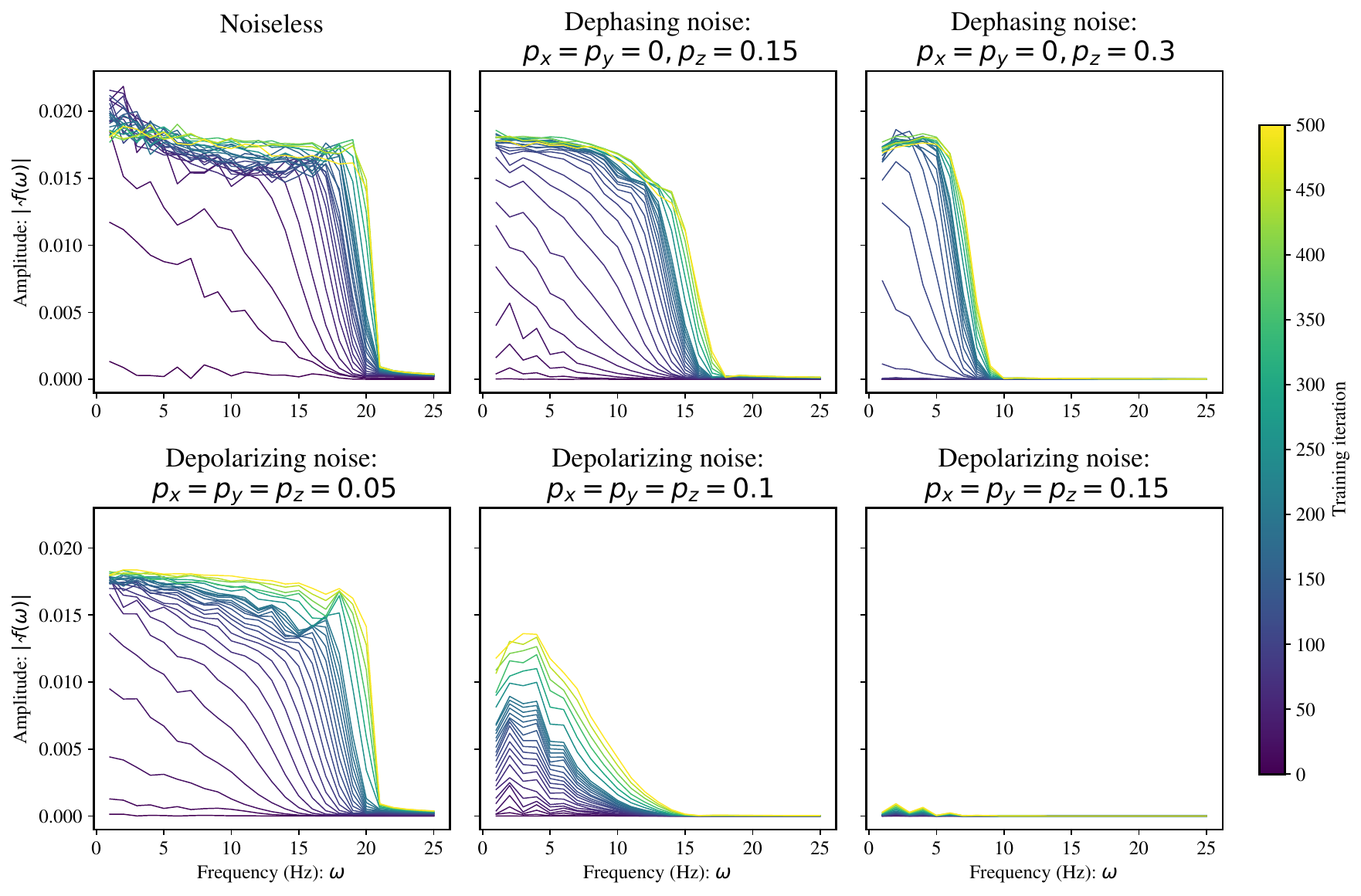}
 \caption{\justifying
 Fourier spectrum of the QNN output \(f_{\bm{\theta}}(x)\) during training under different noise models.  
        Each panel plots the amplitude \(\lvert\hat{f}_{\bm{\theta}}(\omega)\rvert\) versus frequency $\omega$; curves are colour-coded by iteration (dark \(\rightarrow\) early stage, yellow = 500, i.e.\ near convergence).  
        Low-frequency components converge after only a few updates, whereas high-frequency components require many more iterations, evidencing faster learning of low frequencies and noise-induced suppression of high frequencies.
} 
 \label{fig:pic_QNN_Robustness22_a}
\end{figure*}

\begin{figure*}[htbp!]
 \centering
 \includegraphics[width = 0.95\textwidth]{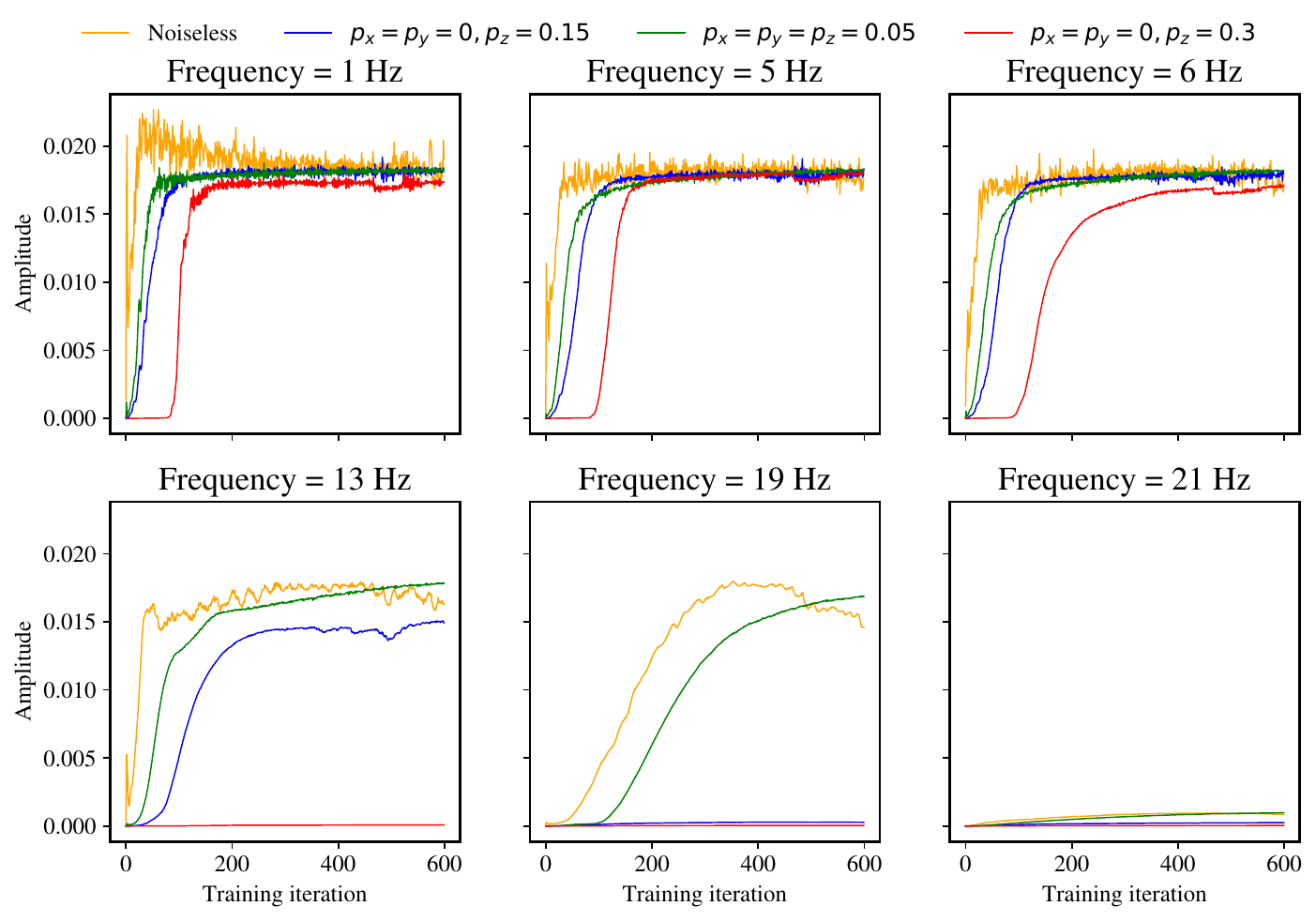}
 \caption{\justifying
Evolution of Fourier amplitudes at selected frequencies (1, 5, 6, 13, 19, and 21 Hz, respectively) during QNN training under various noise models. The curves reflect how noise strength and type affect the learnability of frequency components across training iterations.
} 
 \label{fig:pic_QNN_Robustness22_b}
\end{figure*}

\paragraph*{Noise model architecture.}
We consider noise applied only after the \(R_Z\) gates in the encoding layer, consistent with Sec.~\ref{sec:2.3}, where the rotation angle equals the input \(x\).
To study both dephasing and depolarizing noise in a unified way and to enable an easy comparison, we model them using the Pauli noise channel:
$\mathcal{E}_{\text{Pauli}}(\rho) = (1 - p_x - p_y - p_z)\rho + p_x X\rho X + p_y Y\rho Y + p_z Z\rho Z$.
When \(p_x = p_y = 0\), this reduces to dephasing noise, and when \(p_x = p_y = p_z > 0\), it corresponds to depolarizing noise. 

Furthermore, in Supp.~\ref{appendix:application1_numerical}, we present numerical simulations in a more hardware-relevant setting where noise is inserted after all quantum gates. We apply gate-specific noise: distinct types and strengths immediately after the input-encoding rotation gates, and, for all remaining gates, depolarizing noise with rate $0.001$ for single-qubit gates and $0.01$ for two-qubit gates. This scenario closely reflects the training of noisy QNNs on quantum hardware without quantum error mitigation. We examine how such noise influences the evolution of the QNN's frequency-domain behavior over the course of training. The results also show that, even in the presence of noise, QNNs can effectively learn the low-frequency components of the target function.  These findings suggest that, in certain regimes, the impact of noise on QNN training may be less severe than anticipated, and full-scale quantum error mitigation on every gate may not be strictly necessary.
\subsubsection{Numerical Simulation Results}

The numerical simulation results are presented in Fig.~\ref{fig:pic_QNN_Robustness22_a} and Fig.~\ref{fig:pic_QNN_Robustness22_b}.
Fig.~\ref{fig:pic_QNN_Robustness22_a} tracks the Fourier amplitudes  
\(|\hat{f}_{\bm{\theta}}(\omega)|\) of the QNN output \(f_{\bm{\theta}}(x)\) over \(500\) full-batch iterations  
(one epoch per update).  
Within each panel, darker curves correspond to early iterations; the yellow curve marks iteration 500.

\paragraph*{Noiseless case.}
In the top-left panel, low-frequency modes (\(\omega\!\le\!10\,\mathrm{Hz}\)) reach amplitudes \(0.015\text{-}0.020\) in under \(100\) iterations, whereas high-frequency modes converge far more slowly, in accord with Theorem~\ref{thm:main_theorem1}.

\paragraph*{Dephasing noise.}
Introducing dephasing after the encoding gates with rates  
\((p_x,p_y,p_z)=(0,0,0.15)\) and \((0,0,0.30)\) (centre and right, first row) retards convergence, most strongly at high frequencies.  
At \(\omega=20\,\mathrm{Hz}\) the final amplitude drops from \(0.0175\) (noiseless) to \(0.015\) for \(p_z=0.15\) and is fully suppressed for \(p_z=0.30\).  
These observations confirm Theorem~\ref{thm:frequency_suppression}: dephasing noise imposes an exponential attenuation on high-frequency components while leaving the low-frequency sector essentially unaffected. Notably, for frequencies $\leq 10$ Hz, the final converged amplitudes remain comparable across all three cases, suggesting that dephasing noise does not significantly hinder the QNN's ability to learn the low-frequency components of the target function $f(x)$: it only slows down the convergence.

Building on the conclusion of Theorem~\ref{thm:main_theorem1}, which demonstrates that QNNs predominantly capture low-frequency components during the initial phase of training, our results suggest that QNNs with dephasing noise after rotation gates remain effective in this early stage. Moreover, if the target function primarily consists of low-frequency modes, which is common in the task with classical data~\cite{Field:87,_1994,olshausen1996emergence,simoncelli2001natural}, or the training data contains high-frequency noise, such noisy QNNs may even outperform noiseless counterparts by selectively learning relevant features while being inherently resistant to high-frequency perturbations.

\paragraph*{Depolarizing noise.}
Next, we examine the second-row subfigures in Fig.~\ref{fig:pic_QNN_Robustness22_a}, which depict results under depolarizing noise. Similar to dephasing noise, the presence of depolarizing noise slows down convergence. However, unlike dephasing, the convergence degradation appears approximately uniform across the frequency spectrum. This observation is consistent with prior studies~\cite{aharonov2023polynomial,PhysRevLett.133.120603,fontana2025classical}. 
These works indicate that, under depolarizing noise, suppression depends on the Hamming weight rather than on the frequency.

Fig.~\ref{fig:pic_QNN_Robustness22_b} presents the convergence of amplitude for the selected frequencies under various noise models and noise levels. Each curve corresponds to a specific noise condition, with amplitude plotted as a function of training iteration.
At 1 Hz and 5 Hz (the first two subfigures), the amplitudes converge to approximately 0.02 regardless of noise, indicating that QNNs can successfully learn low-frequency components even under noise. For dephasing noise with rate $p_z = 0.3$, we observe a marked slowdown in convergence as frequency increases from 1 Hz to 6 Hz. At 6 Hz, the amplitude in the noisy case is already lower than in the noiseless case, and by 13 Hz, the amplitude remains at zero throughout training, indicating the QNN fails to learn this component. The same holds for 19 Hz.
For dephasing noise with $p_z = 0.15$, we observe similar behavior: slower convergence at 13 Hz with reduced amplitude, and complete suppression at 19 Hz. These results underscore the exponential suppression of higher-frequency components under dephasing noise, as stated in Lemma~\ref{lemma:frequency_suppression}. Moreover, the larger the noise rate, the more significant the suppression effect.
To more comprehensively illustrate the convergence behavior of different frequency components, we present additional plots for each frequency from 1~Hz to 21~Hz in the Supp.~\ref{appendix:application1_numerical} Fig.~\ref{fig:whole_picture_MPO_QEMappendix}.

\section{Conclusion and Outlook}
\label{sec:Discussion}

In this work, we have developed a unified frequency-domain framework for analyzing the training dynamics of both classical and quantum neural networks. On the theoretical side, we proved a general F-principle theorem on smooth data manifolds, establishing that under gradient-flow dynamics the low-frequency component of the training error decays strictly faster than the total error. This result provides a rigorous foundation for widely observed empirical phenomena in classical deep learning and quantum learning, and bridges learning on Euclidean data spaces with learning on quantum state manifolds.

We further incorporated realistic noise effects by analyzing parameterized quantum circuits subject to axis-aligned single-term Pauli channels and depolarizing noise. Using a Pauli path-integral representation, we demonstrated that these noise channels act as spectral filters, exponentially suppressing high-frequency Fourier modes while largely preserving low-frequency structure. Within this noisy setting, we showed that the gradient of the low-frequency loss continues to dominate that of the full loss, confirming that the unified F-principle remains operative beyond the idealized noiseless regime. Moreover, the same spectral structure enables a frequency-truncation scheme that approximates noisy expectation values with a polynomial-time classical algorithm and controlled average-case error, thereby directly linking noise-induced spectral filtering to classical simulability.

Taken together, our results indicate that the F-principle provides a unifying and predictive perspective on how both classical and quantum neural networks learn in noiseless and noisy environments. This frequency-based viewpoint suggests several promising directions for future research, including the design of ansatz circuits and data-encoding strategies tailored to specific frequency bands, a deeper exploration of the interplay between spectral bias, generalization, and quantum advantage, and extensions of our analysis beyond idealized gradient flow to more realistic optimization dynamics and hardware-level noise models. More broadly, clarifying the respective roles of low- and high-frequency components in learnability and simulability may offer valuable guidance for identifying learning tasks well suited to near-term quantum devices, and for delineating the regimes in which genuine quantum advantage can be expected.

\begin{acknowledgments}

We thank Guoliang Yu and Xiaodie Lin for valuable discussions. Z.L. and Z.W. were supported by the Beijing Natural Science Foundation Key Program (Grant No.~Z220002).
Z.L. was supported by NKPs (Grant No.~2020YFA0713000). R.L. and Z.L. were supported by BMSTC and ACZSP (Grant No.~Z221100002722017). 
R.Z. and Z.W. were supported by the National Natural Science Foundation of China (Grant Nos.~62272259 and 62332009).
W.L. and D.-L.D. are supported by the National Natural Science Foundation of China (Grant No.~T2225008), the Quantum Science and Technology-National Science and Technology Major Project (Grant No.~2021ZD0302203), the Tsinghua University Dushi Program, and the Shanghai Qi Zhi Institute Innovation Program SQZ202318.

\end{acknowledgments}

\bibliography{ref}

\clearpage
\widetext

\appendix
\section*{Supplementary Material}
\renewcommand{\thesection}{\Roman{section}}
\renewcommand{\appendixname}{Supplementary Material}

% ****** Begin of Appendix ******

\section{Rigorous Statement for Theorem 1}\label{sec:append:new_statement}

Let \((\mathcal{M},g)\) be a connected \(d\)-dimensional Riemannian manifold with volume element \(\rm d\boldsymbol{x}\).
All functions appearing in what follows are taken in the Hilbert space \(L^{2}(\mathcal{M})\) and are assumed to possess \(m\) square-integrable weak derivatives for some fixed integer \(m\ge 1\); that is, they lie in the Sobolev space \(H^{m}(\mathcal{M})\).
(The precise value of \(m\) and the accompanying regularity bounds are spelled out later in Assumption~\ref{assumption1}.)
Denote by \(\Delta_{\boldsymbol{x}}\) a positive operator on \(\mathcal{M}\), realized as a self-adjoint operator on \(L^{2}(\mathcal{M})\).
By the spectral theorem\cite{Hall2013QuantumTheory}, \(\Delta_{\boldsymbol{x}}\) is associated with a projection-valued measure \(\mu\); for a Borel set \(I\subset\mathbb{R}_{\ge 0}\) we write \(E_{I}:=\mu(I)\), and in particular \(E_{\lambda}:=E_{[0,\lambda]}\).

The training error \(g_{\boldsymbol{\theta}}=f_{\boldsymbol{\theta}}-y\) admits the spectral decomposition  
\[
  g_{\boldsymbol{\theta}}(\boldsymbol{x})
      \;=\;
      \int_{\xi\ge 0} \rm d\mu(\xi)\,g_{\boldsymbol{\theta}}(\boldsymbol{x})
      \;=\;
      \int_{\xi\ge 0} \rm dg_{\boldsymbol{\theta}}(\xi,\boldsymbol{x}),
\]
with
\(
  \Delta_{\boldsymbol{x}}g_{\boldsymbol{\theta}}(\boldsymbol{x})
    =\int_{\xi\ge 0}\xi\,dg_{\boldsymbol{\theta}}(\xi,\boldsymbol{x}).
\)

Define the full loss and its low-frequency truncation by  
\begin{align}
    \mathcal{L}(\boldsymbol{\theta})& := \int_\mathcal{M}\|f_{\boldsymbol{\theta}}(\boldsymbol{x})-y(\boldsymbol{x})\|^2 \rm d\boldsymbol{x} =\langle g_{\boldsymbol{\theta}}(\boldsymbol{x}),g_{\boldsymbol{\theta}}(\boldsymbol{x})\rangle_{L^{2}(\mathcal{M})} = \langle g_{\boldsymbol{\theta}}(\boldsymbol{x}),\int_\xi ~d\mu(\xi)g_{\boldsymbol{\theta}}(\boldsymbol{x})\rangle_{L^{2}(\mathcal{M})}\\ &=\int_\xi \rm ~d\langle g_{\boldsymbol{\theta}}(\boldsymbol{x}), \mu(\xi)g_{\boldsymbol{\theta}}(\boldsymbol{x})\rangle_{L^{2}(\mathcal{M})}= \int_\xi \rm ~d\langle \mu(\xi)g_{\boldsymbol{\theta}}(\boldsymbol{x}), \mu(\xi)g_{\boldsymbol{\theta}}(\boldsymbol{x})\rangle_{L^{2}(\mathcal{M})}=\int_\xi \rm ~d\langle g_{\boldsymbol{\theta}}(\xi, \boldsymbol{x}), g_{\boldsymbol{\theta}}(\xi, \boldsymbol{x})\rangle_{L^{2}(\mathcal{M})} \;,
\end{align}

and similarly
$$
\mathcal{L}_{\lambda}(\boldsymbol{\theta}) = \langle g_{\boldsymbol{\theta}},E_{\lambda} g_{\boldsymbol{\theta}}\rangle _{L^{2}(\mathcal{M})}=\int_{\xi\in[0, \lambda]} \rm ~d\langle  g_{\boldsymbol{\theta}}(\xi, \cdot), g_{\boldsymbol{\theta}}(\xi, \cdot)\rangle_{L^{2}(\mathcal{M})} \;.
$$
Suppose the training dynamics of $\boldsymbol{\theta}$ are 

\begin{equation}\label{dynamics}
    \left\{\begin{array}{l}
\frac{\mathrm{d} \boldsymbol{\theta}}{\mathrm{~d} t}=-\nabla_{\boldsymbol{\theta}}\mathcal{L}(\boldsymbol{\theta}) \\
\boldsymbol{\theta}(0)=\boldsymbol{\theta}_0
\end{array}\right.
\end{equation}

\begin{definition}[Sobolev space \(H^{m}(\mathcal{M})\)]
Let \((\mathcal{M},g)\) be a \(d\)-dimensional Riemannian manifold with Levi-Civita connection \(\nabla\) and volume element \(d\boldsymbol{x}\).
For an integer \(m\ge 0\), the Sobolev space \(H^{m}(\mathcal{M})\) is defined as
\[
    H^{m}(\mathcal{M})
      \;:=\;
      \Bigl\{
          f\in L^{2}(\mathcal{M})
          \;\Big|\;
          \nabla^{\alpha} f\in L^{2}(\mathcal{M})
          \text{ for every multi-index } \alpha
          \text{ with }|\alpha|\le m
      \Bigr\},
\]
where \(\nabla^{\alpha}\) denotes the weak (distributional) covariant derivative of order \(|\alpha|\).
Equipped with the norm
\[
    \|f\|_{H^{m}(\mathcal{M})}^{2}
      \;:=\;
      \sum_{|\alpha|\le m}
      \int_{\mathcal{M}}
          \bigl|\nabla^{\alpha}f(\boldsymbol{x})\bigr|^{2}\,d\boldsymbol{x},
\]
it is easy to show that \(H^{m}(\mathcal{M})\) is a Hilbert space.
\end{definition}

\begin{assumption}\label{assumption1}\leavevmode
\begin{enumerate}[label=(A\arabic*)]

\item\label{A1}
 {Manifold.}
\(\mathcal{M}\) is a smooth \(d\)-dimensional Riemannian manifold endowed with the volume measure \(d\boldsymbol{x}\).

\item\label{A2}
 {Sobolev regularity.}
For some integer \(m\ge 1\) the error \(g_{\boldsymbol{\theta}}=f_{\boldsymbol{\theta}}-y\) belongs to \(H^{m}(\mathcal{M})\) and
\(
  \Delta_{\boldsymbol{x}}^{m} g_{\boldsymbol{\theta}},\;
  \Delta_{\boldsymbol{x}}^{m}\Delta_{\boldsymbol{\theta}} g_{\boldsymbol{\theta}}
  \in L^{2}(\mathcal{M})
\)
for all \(\boldsymbol{\theta}\) along the training trajectory.

\item\label{A3}
 {Energy bounds.}  %
There exist constants \(C_{1},C_{2}>0\), independent of $m$, such that for all \(\boldsymbol{\theta}(t)\)
\begin{equation}\label{eq: C_1}
    \left|\left\langle\Delta_{\boldsymbol{x}}^m\bigl(g_{\boldsymbol{\theta}}( \boldsymbol{x})\bigr),g_{\boldsymbol{\theta}}(\boldsymbol{x})\right\rangle_{L^{2}(\mathcal{M})}\;\right| = \left|\int_\xi \xi^m~\rm d\left\langle g_{\boldsymbol{\theta}}(\xi, \boldsymbol{x}),g_{\boldsymbol{\theta}}(\xi, \boldsymbol{x})\right\rangle_{L^{2}(\mathcal{M})}\;\right| \le C_1
\end{equation}
 and
 \begin{equation}\label{eq: C_2}
\left|\left\langle\Delta_{\boldsymbol{x}}^m\Delta_{\boldsymbol{\theta}}\bigl(g_{\boldsymbol{\theta}}( \boldsymbol{x})\bigr),g_{\boldsymbol{\theta}}(\boldsymbol{x})\right\rangle_{L^{2}(\mathcal{M})}\;\right| = \left|\int_\xi \xi^m~\rm d\left\langle \Delta_{\boldsymbol{\theta}} g_{\boldsymbol{\theta}}(\xi, \boldsymbol{x}),g_{\boldsymbol{\theta}}(\xi, \boldsymbol{x})\right\rangle_{L^{2}(\mathcal{M})}\;\right| \le C_2.
\end{equation}

\item\label{A4}
 {Gradient regularity.}  %
The map \(t\mapsto\nabla_{\boldsymbol{\theta}}\mathcal{L}(\boldsymbol{\theta}(t))\) is continuous on \([0,T]\), and
\(\nabla_{\boldsymbol{\theta}}\mathcal{L}\) is Lipschitz in \(\boldsymbol{\theta}\) on a full-measure subset of parameter space.

\item\label{A5}
 {Non-degenerate training point.}
For some \(T>0\), \(\bigl|\nabla_{\boldsymbol{\theta}}\mathcal{L}(\boldsymbol{\theta}(T))\bigr|>0\).

\end{enumerate}
\end{assumption}

\begin{theorem}[Rigorous unified F-principle]\label{thm:rigorous}
Under Assumptions~\ref{A1}-\ref{A5} there exists a constant \(\beta>0\) such that, for every \(t\in[0,T]\),
    $$
\left|1-\frac{\rm d\mathcal{L}_\lambda(\boldsymbol{\theta}(t))/\rm dt}{\rm d\mathcal{L}(\boldsymbol{\theta}(t))/\rm dt}\right|\leq \beta\lambda^{-m} \quad \text{and}\quad \left|\frac{\mathrm{d}\mathcal{L}_\lambda(\boldsymbol{\theta}(t))/\mathrm{d}t}{\mathrm{d}\mathcal{L}(\boldsymbol{\theta}(t))/\mathrm{d}t}\right|\geq 1- \beta\lambda^{-m}
$$
\end{theorem}

\begin{proof}
The high-frequency part of dynamics for the loss function is:
$$
\begin{aligned}
\frac{\mathrm{d} \mathcal{L}(\boldsymbol{\theta})}{\mathrm{d} t}-\frac{\rm d\mathcal{L}_\lambda(\boldsymbol{\theta}(t))}{\rm dt} 
& =\left(\int_{\xi\in [\lambda, \infty]}\rm ~d\nabla_{\boldsymbol{\theta}}\langle  g_{\boldsymbol{\theta}}(\xi, \cdot), g_{\boldsymbol{\theta}}(\xi, \cdot)\rangle_{L^{2}(\mathcal{M})}\right)\cdot \frac{\rm d\boldsymbol{\theta}}{\rm dt} \\
& =-\left(\int_{\xi\in [\lambda, \infty]}\rm ~d\nabla_{\boldsymbol{\theta}}\langle  g_{\boldsymbol{\theta}}(\xi, \cdot), g_{\boldsymbol{\theta}}(\xi, \cdot)\rangle_{L^{2}(\mathcal{M})}\right) \cdot \nabla_{\boldsymbol{\theta}} \mathcal{L}(\boldsymbol{\theta})
\end{aligned}
$$

The dynamics for the total loss function are
$$
\frac{\mathrm{d} \mathcal{L}(\boldsymbol{\theta})}{\mathrm{d} t}=-\left|\nabla_{\boldsymbol{\theta}} \mathcal{L}(\boldsymbol{\theta})\right|^2
$$

Therefore
$$
\begin{aligned}
\left|1-\frac{\rm d\mathcal{L}_\lambda(\boldsymbol{\theta}(t))/\rm dt}{\rm d\mathcal{L}(\boldsymbol{\theta}(t))/\rm dt}\right| & =\frac{\left|\rm d\mathcal{L}(\boldsymbol{\theta}(t))/\rm dt -\mathrm{d} \mathcal{L}_\lambda(\boldsymbol{\theta}(t)) / \mathrm{d} t\right|}{\left|\mathrm{d} \mathcal{L} / \mathrm{d} t\right|} \\
& \leq 
\frac{\left| \int_{\xi\in [\lambda, \infty]}\rm ~d\nabla_{\boldsymbol{\theta}}\langle  g_{\boldsymbol{\theta}}(\xi, \cdot), g_{\boldsymbol{\theta}}(\xi, \cdot)\rangle_{L^{2}(\mathcal{M})}\right|\cdot\left|\nabla_{\boldsymbol{\theta}} \mathcal{L}(\boldsymbol{\theta})\right|}{\left|\nabla_{\boldsymbol{\theta}} \mathcal{L}(\boldsymbol{\theta})\right|^2} \\
& \leq
\frac{\int_{\xi\in [\lambda, \infty]}\rm ~d\left|\nabla_{\boldsymbol{\theta}}\langle  g_{\boldsymbol{\theta}}(\xi, \cdot), g_{\boldsymbol{\theta}}(\xi, \cdot)\rangle_{L^{2}(\mathcal{M})}\right|}{\left|\nabla_{\boldsymbol{\theta}} \mathcal{L}(\boldsymbol{\theta})\right|}.
\end{aligned}
$$

Notice that $\lambda\leq \xi$ for all $\xi\in  [\lambda, \infty]$, therefore
$$
\begin{aligned}
\int_{\xi\in [\lambda, \infty]}\rm ~d\left|\nabla_{\boldsymbol{\theta}}\langle  g_{\boldsymbol{\theta}}(\xi, \cdot), g_{\boldsymbol{\theta}}(\xi, \cdot)\rangle_{L^{2}(\mathcal{M})}\right| & \leq \lambda^{-m} \int_{\xi\in [\lambda, \infty]} \rm ~d\left|\nabla_{\boldsymbol{\theta}}\langle \xi^m g_{\boldsymbol{\theta}}(\xi, \cdot), g_{\boldsymbol{\theta}}(\xi, \cdot)\rangle_{L^{2}(\mathcal{M})} \right|
\end{aligned}
$$
 Since 
\begin{align}
    & \left|\nabla_{\boldsymbol{\theta}}\langle \xi^m g_{\boldsymbol{\theta}}(\xi, \cdot), g_{\boldsymbol{\theta}}(\xi, \cdot)\rangle_{L^{2}(\mathcal{M})} \right|\\
    = & \left| \xi^m \overline{g_{\boldsymbol{\theta}}(\xi, \cdot)} \nabla_{\boldsymbol{\theta}} g_{\boldsymbol{\theta}}(\xi, \cdot)+ \xi^m  g_{\boldsymbol{\theta}}(\xi, \cdot)\nabla_{\boldsymbol{\theta}} \overline{g_{\boldsymbol{\theta}}(\xi, \cdot)} \right|\\
    \leq & 
 \left| \xi^m \overline{g_{\boldsymbol{\theta}}(\xi, \cdot)} \nabla_{\boldsymbol{\theta}} g_{\boldsymbol{\theta}}(\xi, \cdot)\right| + \left| \xi^m  g_{\boldsymbol{\theta}}(\xi, \cdot)\nabla_{\boldsymbol{\theta}} \overline{g_{\boldsymbol{\theta}}(\xi, \cdot)} \right|,
\end{align}
it suffices to prove $\int_{\xi\in [\lambda, \infty]} \rm ~d \left|\langle \xi^m \overline{g_{\boldsymbol{\theta}}(\xi, \cdot)} \nabla_{\boldsymbol{\theta}} g_{\boldsymbol{\theta}}(\xi, \cdot)\right|<\infty$ and $\int_{\xi\in [\lambda, \infty]} \rm ~d\left| \xi^m  g_{\boldsymbol{\theta}}(\xi, \cdot)\nabla_{\boldsymbol{\theta}} \overline{g_{\boldsymbol{\theta}}(\xi, \cdot)} \right|<\infty$ respectively. 

Note that $\operatorname{Spec}(\Delta) \subseteq [0, \infty)$. Let $N$ denote the dimension of the parameter vector $\boldsymbol{\theta}$. Then we have:
\begin{align}
&
\int_{\xi\in [\lambda, \infty]} \rm ~d\left|\xi^m \overline{g_{\boldsymbol{\theta}}(\xi, \cdot)} \nabla_{\boldsymbol{\theta}} g_{\boldsymbol{\theta}}(\xi, \cdot)\right|
\\\leq &
\int_{\xi\in [\lambda, \infty]} \rm ~d\sqrt{\langle \xi^\frac{m}{2} \nabla_{\boldsymbol{\theta}} g_{\boldsymbol{\theta}}(\xi, \cdot), \xi^\frac{m}{2} \nabla_{\boldsymbol{\theta}} g_{\boldsymbol{\theta}}(\xi, \cdot)\rangle_{L^{2}(\mathcal{M})}} \cdot\sqrt{N\langle \xi^\frac{m}{2}  g_{\boldsymbol{\theta}}(\xi, \cdot), \xi^\frac{m}{2}  g_{\boldsymbol{\theta}}(\xi, \cdot)\rangle_{L^{2}(\mathcal{M})}} \\
 =&\sqrt{N}\int_{\xi\in [\lambda, \infty]} \rm ~d\sqrt{\langle \xi^{m} \Delta_{\boldsymbol{\theta}} g_{\boldsymbol{\theta}}(\xi, \cdot),  g_{\boldsymbol{\theta}}(\xi, \cdot)\rangle_{L^{2}(\mathcal{M})}}\cdot\sqrt{\langle \xi^{m}  g_{\boldsymbol{\theta}}(\xi, \cdot),  g_{\boldsymbol{\theta}}(\xi, \cdot)\rangle_{L^{2}(\mathcal{M})}}\\
  \leq &\frac{\sqrt{N}}{2}\int_{\xi} \rm ~d\left( \langle \xi^{m} \Delta_{\boldsymbol{\theta}} g_{\boldsymbol{\theta}}(\xi, \cdot),  g_{\boldsymbol{\theta}}(\xi, \cdot)\rangle_{L^{2}(\mathcal{M})}+\langle \xi^{m}  g_{\boldsymbol{\theta}}(\xi, \cdot),  g_{\boldsymbol{\theta}}(\xi, \cdot)\rangle_{L^{2}(\mathcal{M})}\right)\\
 = &\frac{\sqrt{N}}{2}\int_{\xi}\xi^{m} \rm ~d\left( \langle \Delta_{\boldsymbol{\theta}} g_{\boldsymbol{\theta}}(\xi, \cdot),  g_{\boldsymbol{\theta}}(\xi, \cdot)\rangle_{L^{2}(\mathcal{M})}+\langle g_{\boldsymbol{\theta}}(\xi, \cdot),  g_{\boldsymbol{\theta}}(\xi, \cdot)\rangle_{L^{2}(\mathcal{M})}\right)\\
\leq & \frac{\sqrt{N}}{2}\left(\left|\int_\xi \xi^m~\rm d\left\langle g_{\boldsymbol{\theta}}(\xi, \boldsymbol{x}),g_{\boldsymbol{\theta}}(\xi, \boldsymbol{x})\right\rangle_{L^{2}(\mathcal{M})}\;\right| + \left|\int_\xi \xi^m~\rm d\left\langle \Delta_{\boldsymbol{\theta}} g_{\boldsymbol{\theta}}(\xi, \boldsymbol{x}),g_{\boldsymbol{\theta}}(\xi, \boldsymbol{x})\right\rangle_{L^{2}(\mathcal{M})}\;\right|\right)\\
= & \frac{\sqrt{N}}{2}(C_1+C_2),
 \end{align}
and we can prove that $\int_{\xi\in [\lambda, \infty]} \rm ~d\left|\langle \xi^m  g_{\boldsymbol{\theta}}(\xi, \cdot),\nabla_{\boldsymbol{\theta}} g_{\boldsymbol{\theta}}(\xi, \cdot)\rangle_{L^{2}(\mathcal{M})} \right|<\infty$ by the same argument.

 By assumption~\ref{assumption1}(A4),  $\nabla_{\boldsymbol{\theta}} \mathcal{L}(\boldsymbol{\theta}(t))$ is continuous on $t \in[0, T]$. If $\inf _{t \in(0, T]}\left|\nabla_{\boldsymbol{\theta}} \mathcal{L}(\boldsymbol{\theta}(t))\right|=0$, then there is a $t_0 \in[0, T]$ such that $\left|\nabla_{\boldsymbol{\theta}} \mathcal{L}\left(\boldsymbol{\theta}\left(t_0\right)\right)\right|=0$. By the uniqueness of the ordinary differential equation~\ref{dynamics}, $\boldsymbol{\theta}(t)\equiv \boldsymbol{\theta}_0$ is the unique solution, thus we have $\left|\nabla_{\boldsymbol{\theta}} \mathcal{L}(\boldsymbol{\theta}(T))\right|= \left|\nabla_{\boldsymbol{\theta}} \mathcal{L}(\boldsymbol{\theta}_0)\right|=\left|\nabla_{\boldsymbol{\theta}} \mathcal{L}(\boldsymbol{\theta}(t_0))\right|=0$ which contradicts the assumption that $\left|\nabla_{\boldsymbol{\theta}} \mathcal{L}(\boldsymbol{\theta}(T))\right|>0$. 
Hence $\inf _{t \in(0, T]}\left|\nabla_{\boldsymbol{\theta}} \mathcal{L}(\boldsymbol{\theta}(t))\right|>0$. Thus, the constant $\beta:= \frac{\sqrt{N}(C_1 + C_2)}{2\inf_{t \in (0,T]}\left|\nabla_{\boldsymbol{\theta}}\mathcal{L}(\boldsymbol{\theta})\right|}$ satisfies $0 < \beta < +\infty$.
Therefore, we have proved 
\begin{align*}
    \left|1-\frac{\rm d\mathcal{L}_\lambda(\boldsymbol{\theta}(t))/\rm dt}{\rm d\mathcal{L}(\boldsymbol{\theta}(t))/\rm dt}\right|
 & \leq 
\frac{\int_{\xi\in [\lambda, \infty]}\rm ~d\left|\nabla_{\boldsymbol{\theta}}\langle  g_{\boldsymbol{\theta}}(\xi, \cdot), g_{\boldsymbol{\theta}}(\xi, \cdot)\rangle_{L^{2}(\mathcal{M})}\right|}{\left|\nabla_{\boldsymbol{\theta}} \mathcal{L}(\boldsymbol{\theta})\right|}\\
&\leq
\lambda^{-m}\frac{\int_{\xi\in [\lambda, \infty]} \rm ~d\left|\nabla_{\boldsymbol{\theta}}\langle \xi^m g_{\boldsymbol{\theta}}(\xi, \cdot), g_{\boldsymbol{\theta}}(\xi, \cdot)\rangle_{L^{2}(\mathcal{M})} \right|}{\left|\nabla_{\boldsymbol{\theta}} \mathcal{L}(\boldsymbol{\theta})\right|}\\
&\leq
\beta\lambda^{-m},
\end{align*}
and
$$\left|\frac{\rm d\mathcal{L}_\lambda(\boldsymbol{\theta}(t))/\rm dt}{\rm d\mathcal{L}(\boldsymbol{\theta}(t))/\rm dt}\right|
\geq 
1- \left|1-\frac{\rm d\mathcal{L}_\lambda(\boldsymbol{\theta}(t))/\rm dt}{\rm d\mathcal{L}(\boldsymbol{\theta}(t))/\rm dt}\right|
\geq
1- \beta\lambda^{-m}.
$$
\end{proof}

\section{F-principle for QNN}\label{appendix::sec::QNN}

In this section, we address a quantum machine learning task in which the objective is to train a model on a dataset \( \mathcal{S} = \left\{ \rho^{(s)} \right\}_{s=1}^N \) consisting of \( n \)-qubit quantum states. The quantum model is parameterized by a QNN, which is a unitary quantum channel \( \mathcal{C}_\theta \). This channel acts on the input states \( \rho^{(s)} \) according to the transformation
\[
\mathcal{C}_\theta\left( \rho^{(s)} \right) = U(\theta) \rho^{(s)} U^{\dagger}(\theta).
\] Here, $U(\theta)$ is defined as
\begin{equation}
    U(\theta)=\prod_{m=1}^M U_m\left(\theta_m\right)W_m, \quad U_m\left(\theta_m\right)=e^{-i \theta_m H_m}
\end{equation}
where $\theta \in \mathbb{R}^M$ is a vector of trainable parameters, $W_m$ is a generic unitary operator that does not depend on any angle $\theta_m$, typically chosen from Clifford gates. The operators $H_m$ are traceless Hermitian operators drawn from a set of generators $\mathcal{G}$. This allows us to express $\mathcal{C}_\theta$ as a concatenation of $M$ unitary channels
\begin{equation}
\mathcal{C}_\theta=\mathcal{C}_{\theta_M}^M \circ \cdots \circ \mathcal{C}_{\theta_1}^1,
\end{equation}
with $\mathcal{C}_{\theta_m}^{\mathrm{m}}\left(\rho^{(s)}\right)=U_m\left(\theta_m\right) \rho^{(s)} U_m^{\dagger}\left(\theta_m\right)$. Thus, the output of the QNN is a parametrized state
\begin{equation}
\rho_\theta^{(s)}=\mathcal{C}_\theta\left(\rho^{(s)}\right)=\mathcal{C}_{\theta_M}^M \circ \cdots \circ \mathcal{C}_{\theta_1}^1\left(\rho^{(s)}\right)
\end{equation}
Suppose the training set is $\mathcal{D} = \{(x_s, y_s)\}_{s=1}^N \subset \Omega$, where $\Omega$ is the sample space and data in $\mathcal{D}$ is sampled by distributed measure $\mu$ independently.
And the loss function is defined as 
\begin{equation}
\mathcal{L}(\theta)=\frac{1}{N}\sum_{s=1}^N \left(\operatorname{Tr}\left[\mathcal{C}_\theta\left(\rho^{(s)}\right) O_s\right]\right)^2,
\end{equation}
where \( O_s \) is a (potentially) data-instance-dependent measurement.
We will prove that, for sufficiently large training set size \( N \), the gap between \( \mathcal{L}(\theta) \) and \( \int_\Omega \left( \operatorname{Tr}\left[\mathcal{C}_\theta\left(\rho^{(s)}\right) O_s\right] \right)^2 d\mu \) is less than \( \varepsilon \) with probability \( 1 - \delta \). This can be straightforwardly shown using the Hoeffding bound~\cite{Hoeffding1994}:
\[
\operatorname{Pr} \left( \left|
\mathcal{L}(\theta)
-
\int_\Omega \left( \operatorname{Tr}\left[\mathcal{C}_\theta\left(\rho^{(s)}\right) O_s\right] \right)^2 d\mu
\right| > \epsilon \right) 
\leq 2 \exp{\left(-\frac{2 N \epsilon^2}{(b-a)^2}\right)},
\]
where \( a \) and \( b \) represent the lower and upper bounds of \( \left( \operatorname{Tr}\left[\mathcal{C}_\theta\left(\rho^{(s)}\right) O_s\right] \right)^2 \), respectively.
To ensure this inequality holds, it suffices to choose the training set size \( N \) such that
\[
2 \exp{\left(-\frac{2 N \epsilon^2}{(b-a)^2}\right)} \leq \delta.
\]
This implies the requirement
\[
N \geq \frac{(b-a)^2}{2 \epsilon^2} \ln \left( \frac{2}{\delta} \right).
\]
Thus, with high probability, the variational parameters \( \theta \) are trained by minimizing the loss function \( \mathcal{L}(\theta) \), which we assume to be of the form
\[
\mathcal{L}(\theta)  \approx \int_\Omega \left( \operatorname{Tr}\left[ \mathcal{C}_\theta\left(\rho^{(s)}\right) O_s \right] \right)^2 d\mu(\rho^{(s)}).
\]
By considering the quantum state inside the trace map after the quantum neural network on the data \( \rho^{(s)} \) as a function of \( \rho^{(s)} \), the approximate loss function can be written as
\[
\int_\Omega \left( \operatorname{Tr}\left[\mathcal{C}_\theta\left(\rho^{(s)}\right) O_s \right] \right)^2 d\mu(\rho^{(s)}) = 
\int_{\mathbb{CP}^{2^n -1}} \left( \mathbf{1}_\Omega \operatorname{Tr}\left[ \mathcal{C}_\theta\left(\rho^{(s)}\right) O_s \right] \right)^2 d\mu(\rho^{(s)}) 
= \int_{\mathbb{CP}^{2^n -1}} g(\theta,x)^2 d\mu(x),
\]
where \( \mathbf{1}_\Omega \) is the indicator function of the set \( \Omega \), i.e., it takes the value of 1 when \( x \in \Omega \) and 0 otherwise. 

Since the gates and activation functions in quantum neural networks are smooth functions of \( x \) and $\theta$, \( g(\theta,x) \) is therefore an almost-everywhere smooth function on the smooth manifold \( \mathbb{CP}^{2^n -1} \). Consequently, the loss function of the quantum machine learning model we consider in this case satisfies the form and assumptions in Theorem~\ref{thm:main_theorem1}, meaning that the quantum neural networks described above adhere to the frequency principle when the training dynamics satisfy Eq.\ref{eq: train dynamics}. In other words, quantum neural networks first learn the low-frequency components of the quantum data, followed by learning the high-frequency components.

\section{Noise Effect Frequency Components in Parameterized Quantum Circuits}\label{appendix:sec:noise_effect_frequency_components}

In this section, we demonstrate that { Pauli noise } in Eq.~\eqref{eq:generalized_noise_model} introduced in rotation gates, where the Pauli operator in the noise channel aligns with the rotation axis, significantly suppresses the high-frequency components of the circuit. This can be viewed as a generalization of the result presented in Ref.~\cite{fontana2022spectral,fontana2025classical}.
We restate the result as follows:
\begin{lemma}[\cite{fontana2025classical}]\label{lemma:frequency_suppression}
Consider a parameterized quantum circuit  consisting of $L$ rotation gates, where all the rotation gates are the  {$R_Z$ gate} with parameter $\bm{\theta} \in (S^{1})^{L}\equiv \mathbb{T}^{L}$, and the observable is $O$. Denote the ideal expectation value as a function $\langle O(\bm{\theta})\rangle$. Suppose each $R_Z$ gate is followed by a  {dephasing channel} with strength $\gamma$. Then in the frequency domain, the noisy expectation value is given by
\begin{equation}
    \langle O_{\mathrm{noisy}}(\bm{\theta})\rangle
     = \sum_{\bm{\omega} \in \mathbb{Z}^L} (1 - 2\gamma)^{\|\bm{\omega}\|_1} \, \hat{f}(\bm{\omega}) \, e^{i \bm{\omega} \cdot \bm{\theta}},
\end{equation}
where $\hat{f}(\bm{\omega})$ are the Fourier coefficients of the noiseless expectation value
\begin{equation}
    \langle O(\bm{\theta})\rangle = \sum_{\bm{\omega} \in \mathbb{Z}^L} \hat{f}(\bm{\omega}) \, e^{i \bm{\omega} \cdot \bm{\theta}},
\end{equation}
and $\|\bm{\omega}\|_1$ denotes the $\ell_1$-norm of $\bm{\omega}$.
\end{lemma}

We begin by presenting the circuit model and noise model.
Next, we show that if the expectation value of the parameterized quantum circuits is viewed as a function of the rotation angles, then under our noise model, noise suppresses the high-frequency components of this function.
Finally, we prove that the noisy circuit expectation can be well approximated by a low-frequency truncation, and the latter can be computed with polynomial complexity.

\subsection{Circuit Model and Noise Model}\label{app:sec:circuit_model_and_noise_model}

In the current NISQ era, parameterized quantum circuits are widely employed in various near-term quantum algorithms~\cite{kandala2017hardware,farhi2014quantum}.  
A typical $n$-qubit parameterized quantum circuit, denoted as $C(\bm{\theta})$, consists of a sequence of Pauli rotation gates and non-parameterized Clifford gates. Each Pauli rotation gate is written as $e^{-i\frac{\theta}{2} P}$, where $\theta$ is the rotation angle and $P \in \{\mathbb{I}, X, Y, Z\}^{\otimes n}$. Clifford gates are unitary operators that normalize the $n$-qubit Pauli group $\bm{P}_n$:  
\[
\{C \in U_{2^n} \mid C \bm{P}_n C^\dagger = \bm{P}_n\}.
\]

Without loss of generality, we assume that a parameterized quantum circuit takes the form
\begin{equation}
  C(\bm{\theta}) = U_L(\theta_L) \cdots U_1(\theta_1),
\end{equation}
where $\bm{\theta} = (\theta_1, \ldots, \theta_L)$ are the rotation angle and $L$ is the circuit depth. Each unitary $U_i(\theta_i)$ is of the form 
\[
U_i(\theta_i) := e^{-i \theta_i P_i / 2} C_i,
\]
where $P_i \in \{\mathbb{I}, X, Y, Z\}^{\otimes n}$ is a $n$-qubit Pauli operator and $C_i$ is a Clifford gate.

To describe noise effects, we also consider the channel representation of the parameterized quantum circuits. Let $\mathcal{U}_i(\theta_i)$ denote the unitary channel defined by
\[
\mathcal{U}_i(\theta_i)(\rho) = U_i(\theta_i)\rho U_i^\dagger(\theta_i),
\]
where $\rho$ is a quantum state. The full channel $\mathcal{C}(\bm{\theta})$ is then given by the composition:
\begin{equation}
  \mathcal{C}(\bm{\theta}) = \mathcal{U}_L(\theta_L) \circ \cdots \circ \mathcal{U}_1(\theta_1).
\end{equation}
This formulation of parameterized quantum circuits is sufficiently general to encompass most ansatze used in near-term quantum algorithms.

Mostly, the quantum circuit $\mathcal{C}(\bm{\theta})$ is applied to an initial state $\rho$, and what we are interested in is the expectation value of an observable $O$ on the final state, given by
\begin{equation}
  \langle O \rangle = \tr{O \mathcal{C}(\bm{\theta})(\rho)}.
\end{equation}
However, in practice, the quantum circuit is often subject to noise. We use $\tilde{\mathcal{U}}_i(\theta_i)$ to denote the noisy version of the unitary channel $\mathcal{U}_i(\theta_i)$, which can be modeled as a quantum channel. And the unitary channel with noise $\tilde{\mathcal{U}}_i(\theta_i)$ can be decomposed into a unitary channel $\mathcal{U}_i(\theta_i)$ and a noise channel $\mathcal{E}_i$:
\begin{equation}\label{eq:unitary_channel_with_noise}
  \tilde{\mathcal{U}}_i(\theta_i)=\mathcal{E}_i \circ \mathcal{U}_i(\theta_i).
\end{equation}
The noisy quantum circuit $\tilde{\mathcal{C}}(\bm{\theta})$ is then defined as the composition of the noisy unitary channels:
\begin{equation}
    \begin{aligned}
          \tilde{\mathcal{C}}(\bm{\theta}) & =\tilde{\mathcal{U}}_L(\theta_L) \circ \cdots \circ \tilde{\mathcal{U}}_1(\theta_1)  =\mathcal{E}_L \circ \mathcal{U}_L(\theta_L) \circ \cdots \circ \mathcal{E}_1 \circ \mathcal{U}_1(\theta_1).
    \end{aligned}
\end{equation}
The noisy circuit expectation value is denoted as
\begin{equation}\label{eq:noisy_expectation_value}
    \langle O_{\text{noisy}}\rangle=\tr{O \tilde{\mathcal{C}}(\bm{\theta})(\rho) }.
\end{equation}

In our context, we mainly consider two types of noise: 1. Pauli noise in Eq.~\eqref{eq:noise_channel}  introduced in rotation gates, where the Pauli operator in the noise channel aligns with the rotation axis, significantly suppresses the high-frequency components of the circuit. 2. Local depolarizing noise. 

For the noise channel $\mathcal{E}_i$ in Eq.~\eqref{eq:unitary_channel_with_noise}, suppose $U_i(\theta_i) = \exp{-i \frac{\theta_i}{2} P_i} C_i$, then noise channel $\mathcal{E}_i$ satisfies
\begin{equation}\label{eq:noise_channel}
\mathcal{E}_i(\rho) = (1-\gamma) \rho + \gamma P_i \rho P_i^\dagger,
\end{equation}
where $\gamma$ is the noise rate and $P_i \in \{X, Y, Z\}^{\otimes n}$ is the Pauli operator that aligns with the rotation axis of the rotation gate $\exp{-i \frac{\theta_i}{2} P_i}$.

As for local depolarizing noise, we consider the following noise channel on the single qubit:
\begin{equation}\label{eq:noise_channel_dep}
\mathcal{E}_{\mathrm{depol}}(\rho) = (1 - \frac34\gamma)\rho + \frac{\gamma}{4}(X \rho X + Y \rho Y + Z \rho Z),
\end{equation}
which uniformly applies one of the three non-identity Pauli errors with equal probability. It models a complete loss of quantum information with probability $\gamma$.

We { reiterate} that the rotation gates in the parameterized quantum circuits are no longer restricted to \(R_Z\) gates; they are generalized to \(\exp{-i P_I\theta_i/2}\), where \(P_i \in \{I, X, Y, Z\}^{\otimes n}\) denotes an \(n\)-qubit Pauli operator, for \(i = 1,2,\cdots, L\). 
Moreover, the noise on a rotation gate \(\exp{-i P_i\theta_i/2}\) is not limited to dephasing; we adopt the more general Pauli channel in Eq.~\eqref{eq:noise_channel}.

\subsection{Impact of Noise on the Frequency Spectrum of Parameterized Quantum Circuits}\label{app:sec:impact_of_noise_on_frequency_spectrum}

The expectation value of a parameterized quantum circuit can be regarded as a function of the rotation gate angles $\bm{\theta}$. In this section, we analyze the impact of noise-particularly the noise model defined in Eq.~\eqref{eq:noise_channel}-on the frequency spectrum, and demonstrate that noise applied after rotation gates can effectively suppress high-frequency components, resulting in a spectrum dominated by low-frequency features.

To analyze the impact of noise on the frequency spectrum, we introduce the concept of Pauli paths and express the expectation value of a quantum circuit using the Pauli path integral formalism. 
This technique has been widely used for classically simulating quantum circuits~\cite{aharonov2023polynomial,PhysRevLett.133.120603,lerch2024efficient,angrisani2025simulating,fontana2025classical,rudolph2025pauli}.
In our context, this approach provides an intuitive representation of the expectation value as a trigonometric polynomial in the variational parameters $\bm{\theta}$ for both noisy and noiseless circuits.

We begin by analyzing the noiseless case and introduce the definitions of Pauli paths and the Pauli path integral. Although these concepts have been comprehensively discussed in Ref.~\cite{aharonov2023polynomial,PhysRevLett.133.120603,rudolph2025pauli}, we present them here to ensure the completeness of our context, and we use the notation as in Ref.~\cite{zhang2024clifford}.
We consider the expectation value
\begin{equation}\label{eq:exp_val_m}
  \begin{aligned}
      \langle O \rangle =  \tr{{U}_L(\theta_L)  \cdots {U}_1(\theta_1) \rho  {U}_1(\theta_1)^{\dagger} \cdots  {U}_L(\theta_L)^{\dagger} O } 
  \end{aligned}
\end{equation}
where $O \in \{\mathbb{I}, X, Y, Z\}^{\otimes n}$ is an $n$-qubit Pauli operator, $\rho$ denotes the $n$-qubit initial state, and the unitary sequence ${U}_L(\theta_L)  \cdots {U}_1(\theta_1)$ represents the parameterized quantum circuit. Each unitary gate is defined as ${U}_i(\theta_i):=C_i\exp{-i \theta_i P_i / 2} $, where $P_i \in \{\mathbb{I}, X, Y, Z\}^{\otimes n}$ is a Pauli operator and $C_i$ is a Clifford gate.

In the Heisenberg picture, the expectation value in Eq.~\eqref{eq:exp_val_m} can be rewritten as $\langle O \rangle = \operatorname{Tr}(\rho \tilde{O})$, where the Heisenberg-evolved observable is given by
\[
\tilde{O} = {U}_1^{\dagger}(\theta_1) \cdots  {U}_L^{\dagger}(\theta_L) O {U}_L(\theta_L)  \cdots {U}_1(\theta_1).
\]
Therefore, evaluating the expectation value is equivalent to applying Heisenberg evolution to the observable $O$, followed by measurement with respect to the initial state $\rho$.

Specifically, a Heisenberg evolution of a Pauli operator $P'$ under the unitary gate ${U}(\theta):=\exp{-i \theta P / 2}C$ is given by 
\begin{equation}
{U}(\theta)^{\dagger} P' {U}(\theta)= C^\dagger\exp{i \theta P / 2}  P^\prime \exp{-i \theta P / 2}C.
\end{equation}
The action of a Pauli rotation on a Pauli operator results in
\begin{equation}
e^{i \theta P / 2} P^\prime e^{-i \theta P / 2} = 
\begin{cases}
P^\prime, & [P, P^\prime] = 0, \\
\cos(\theta) P^\prime + i \sin(\theta) P P^\prime & \{P, P^\prime\} = 0.
\end{cases} \label{eq:evolution}
\end{equation}

When $P'$ anticommutes with $P$, the operator $P'$ is transformed into a linear combination of Pauli operators. We refer to these resulting Pauli terms as distinct  {Pauli paths}. Thus, the expression $\cos(\theta)\, P' + i \sin(\theta)\, P P'$ can be interpreted as a sum over two Pauli paths. Since both $P$ and $i P P'$ are Pauli operators, the overall evolution becomes:
\begin{equation}
  {U}(\theta)^{\dagger} P' {U}(\theta) = C^\dagger \left( \cos(\theta)\, P' + i \sin(\theta)\, P P' \right) C,
\end{equation}
where $C^\dagger P' C$ and $i C^\dagger P P' C$ are both Pauli operators. Therefore, in the anticommuting case, applying $U(\theta)$ to a Pauli operator results in a linear combination of two Pauli operators, each weighted by a first-order trigonometric polynomial.

In contrast, when $P'$ commutes with $P$, the evolution yields a single Pauli operator:
\begin{equation}
    {U}(\theta)^{\dagger} P' {U}(\theta) = C^\dagger P' C.
\end{equation}
To formalize this recursive evolution process, we introduce the notion of a {Pauli path}.

A Pauli path is defined as a sequence $s = (s_0, s_1, \dots, s_L) \in \bm{P}_n^{L+1}$, where $\bm{P}_n = \left\{ \sfrac{\mathbb{I}}{\sqrt{2}}, \sfrac{X}{\sqrt{2}}, \sfrac{Y}{\sqrt{2}}, \sfrac{Z}{\sqrt{2}} \right\}^{\otimes n}$ denotes the set of normalized $n$-qubit Pauli words.

The path $s$ records the evolution of an observable under the Heisenberg picture. Initially, $s_L$ is set to the normalized observable $\sfrac{O}{(\sqrt{2})^n}$, indicating that the evolution begins from $O$. Then $s_{L-1}$ represents the possible Pauli operator obtained by evolving $s_L$ through the unitary gate $U_L(\theta_L) = C_L \exp{-i \theta_L P_L / 2}$. That is,
\begin{enumerate}[label=(\alph*)]
  \item If $P_L$ commutes with $C_L^\dagger s_L C_L$, then $s_{L-1} = C_L^\dagger s_L C_L$.
  \item If $P_L$ anticommutes with $C_L^\dagger s_L C_L$, then $s_{L-1}$ can be either $C_L^\dagger s_L C_L$ or $iP_L C_L^\dagger s_L C_L$, corresponding to the two possible evolution paths.
\end{enumerate}
Therefore, in the commutation case, 
$\Tr{s_i\mathcal{U}_i(\theta_i)(s_{i-1})}$ is constant. In anticommutation case, if $s_{L-1} = C_L^\dagger s_L C_L$, then $\Tr{s_i\mathcal{U}_i(\theta_i)(s_{i-1})}$ contains the $\cos(\theta_i)$ term and if $s_{L-1} = iP_L C_L^\dagger s_L C_L$, then $\Tr{s_i\mathcal{U}_i(\theta_i)(s_{i-1})}$ contains the $\sin(\theta_i)$ term.
Therefore, in both cases, $\Tr{s_i\mathcal{U}_i(\theta_i)(s_{i-1})}$ is either a constant or a constant multiplied by $\sin(\theta_i)$ or $\cos(\theta_i)$.

This process continues recursively until $s_0$ is determined. The complete sequence $s$ thus describes an evolution path of the observable $O$ through the circuit. While the overall global phase of each Pauli product is omitted in the path representation, it is captured by the following path weight function (sometimes we also call it the coefficient of path $s$):
\begin{equation}
  \begin{aligned}\label{eq:f00}
    f(s,\bm{\theta})=&\Tr{Os_L}\left(\prod_{i=1}^{L}\Tr{s_iU_i(\theta_i) s_{i-1}U_i(\theta_i)^\dagger}\right).
  \end{aligned}
\end{equation}
If $s$ is a valid path generated through the Heisenberg evolution of $O$, then $f(s, \bm{\theta})$ will be a coefficient expressed as a product of sine and cosine functions. Otherwise, $f(s, \bm{\theta}) = 0$.

Using the Pauli path integral, the Heisenberg-evolved observable can be expressed as
\begin{equation}\label{eq:heisenberg_evolved_observable}
  \begin{aligned}
    \tilde{O} &= {U}_1(\theta_1)^{\dagger} \cdots  {U}_L(\theta_L)^{\dagger} O {U}_L(\theta_L)  \cdots {U}_1(\theta_1)\\
    &=\sum_{s_0\in \bm{P}_n} \tr{{U}_1(\theta_1)^{\dagger} \cdots  {U}_L(\theta_L)^{\dagger} O {U}_L(\theta_L)  \cdots {U}_1(\theta_1)s_0} s_0\\
    &=\sum_{s_0\in \bm{P}_n} \tr{{U}_2(\theta_2)^{\dagger} \cdots  {U}_L(\theta_L)^{\dagger} O {U}_L(\theta_L)  \cdots {U}_1(\theta_1)s_0{U}_1(\theta_1)^{\dagger}} s_0\\
    &=\sum_{s_0,s_1\in \bm{P}_n} \tr{{U}_2(\theta_2)^{\dagger} \cdots  {U}_L(\theta_L)^{\dagger} O {U}_L(\theta_L)  \cdots {U}_2(\theta_2)s_1} \tr{s_1{U}_1(\theta_1)s_0{U}_1(\theta_1)^{\dagger}} s_0\\
    &\vdots\\
    &=\sum_{s\in \bm{P}^{L+1}_n}\Tr{Os_L}\left(\prod_{i=1}^{L}\Tr{s_iU_i(\theta_i) s_{i-1}U_i(\theta_i)^\dagger}\right) s_0\\
    &=\sum_{s\in \bm{P}^{L+1}_n} f(s,\bm{\theta}) s_0,
  \end{aligned}
\end{equation}
where the second and fourth equalities hold because $\bm{P}_n = \left\{ \sfrac{\mathbb{I}}{\sqrt{2}}, \sfrac{X}{\sqrt{2}}, \sfrac{Y}{\sqrt{2}}, \sfrac{Z}{\sqrt{2}} \right\}^{\otimes n}$ forms an orthonormal basis.

The expectation value of the observable $O$ can then be written as
\begin{equation}
  \begin{aligned}
    \langle O \rangle &= \tr{\rho \tilde{O}} =\sum_{s\in \bm{P}^{L+1}_n} f(s,\bm{\theta}) \tr{\rho s_0}.
  \end{aligned}
\end{equation}

For the noise circuit in Eq.~\eqref{eq:noisy_expectation_value} with noise model in Eq.~\eqref{eq:noise_channel} and Eq.~\eqref{eq:noise_channel_dep}, in Heisenberg evolution the expectation value can be expressed as
$$
\langle O_{\text{noisy}}\rangle = \tr{\rho\tilde{O}_{\text{noisy}}}
$$
where
$$
\tilde{O}_{\text{noisy}} = \mathcal{U}_1^\dagger \circ \mathcal{E}_1 \circ \cdots \circ \mathcal{U}_L^\dagger \circ \mathcal{E}_L (O).
$$
Note that $\mathcal{U}_i^\dagger$ is the adjoint of $\mathcal{U}_i$, and since $\mathcal{U}_i$ is a unitary channel, $\mathcal{U}_i^\dagger$ is the inverse of $\mathcal{U}_i$. 
And the noise channel $\mathcal{E}_i$ is a Pauli channel, therefore it is self-adjoint.

In noise case, Eq.~\eqref{eq:heisenberg_evolved_observable} can be written as
\begin{equation}\label{eq:heisenberg_evolved_observable_noise}
  \begin{aligned}
    \tilde{O}_{\text{noisy}} &=  \mathcal{U}_1^\dagger \circ \mathcal{E}_1 \circ \cdots \circ \mathcal{U}_L^\dagger \circ \mathcal{E}_L (O)\\
    &=\sum_{s_0\in \bm{P}_n} \tr{ \mathcal{U}_1^\dagger \circ \mathcal{E}_1 \circ \cdots \circ \mathcal{U}_L^\dagger \circ \mathcal{E}_L (O)s_0} s_0\\
    &=\sum_{s_0\in \bm{P}_n} \tr{ \mathcal{U}_2^\dagger \circ \mathcal{E}_2 \circ \cdots \circ \mathcal{U}_L^\dagger \circ \mathcal{E}_L (O) \mathcal{E}_1\circ \mathcal{U}_1(s_0)} s_0\\
    &=\sum_{s_0,s_1\in \bm{P}_n} \tr{\mathcal{U}_2^\dagger \circ \mathcal{E}_2 \circ \cdots \circ \mathcal{U}_L^\dagger \circ \mathcal{E}_L (O)s_1} \tr{s_1\mathcal{E}_1\circ \mathcal{U}_1(s_0)} s_0\\
    &\vdots\\
    &=\sum_{s\in \bm{P}^{L+1}_n}\Tr{Os_L}\left(\prod_{i=1}^{L}\Tr{s_i\mathcal{E}_i\circ \mathcal{U}_i(s_{i-1})}\right) s_0\\
    &=\sum_{s\in \bm{P}^{L+1}_n}\Tr{Os_L}\left(\prod_{i=1}^{L}\Tr{\mathcal{E}_i(s_i)\mathcal{U}_i(s_{i-1})}\right) s_0\\
    &=\sum_{s\in \bm{P}^{L+1}_n} \tilde f(s,\bm{\theta}) s_0,
  \end{aligned}
\end{equation}
where $\tilde f(s,\bm{\theta}) = \Tr{Os_L}\left(\prod_{i=1}^{L}\Tr{\mathcal{E}_i(s_i)\mathcal{U}_i(s_{i-1})}\right)$ is weight (or coefficient) of the Pauli path $s$ with noise and the second last equality holds because the Pauli noise channel $\mathcal{E}_i$ is self-adjoint with respect to the Hilbert-Schmidt inner product.

We show that for each Pauli path $s$, the corresponding Pauli coefficient $f(s,\bm{\theta})$ is scaled by a factor of $(1 - 2\gamma)^{\#\mathrm{trig}(s)}$ in noise model of Eq.~\eqref{eq:noise_channel}, where $\#\mathrm{trig}(s)$ denotes the degree of $f(s,\bm{\theta})$ as a trigonometric polynomial. 
This result directly implies that, when viewing the expectation value of a parameterized quantum circuit as a function of its parameters, the introduction of noise exponentially suppresses the high-frequency components in the frequency domain-that is, the higher the frequency, the stronger the suppression, and the suppression factor is exponential with frequency.

We take Pauli path $s$ with non-zero coefficient $\tilde f(s,\bm{\theta})$ as an example to illustrate this point.
Recall that, in the noise case, the Pauli path coefficient $\tilde f(s,\bm{\theta})$ can be expressed as $$
\Tr{Os_L}\left(\prod_{i=1}^{L}\Tr{\mathcal{E}_i(s_i)\mathcal{U}_i(s_{i-1})}\right).
$$
For each $i = 1,2,\cdots L$, if $\mathcal{U}_i$ is the Clifford gate (i.e. $U_i(\theta_i) = \exp{-i\theta_i I^{\otimes n}/2}C_i$), then $\Tr{\mathcal{E}_i(s_i)\mathcal{U}_i(s_{i-1})}$ is a constant, which does not depend on the parameters $\bm{\theta}$. 
And since the noise is only introduced in the rotation gate, $\mathcal{E}_i$ is the identity channel, $\Tr{\mathcal{E}_i(s_i)\mathcal{U}_i(s_{i-1})} = \Tr{s_i\mathcal{U}_i(s_{i-1})}$, which means noise does not affect this term.

If $\mathcal{U}_i$ is a non-Clifford gate, then it can be written as
\[
U_i = \exp{-i\theta_i P_i/2}C_i.
\]
Since $P_i$ is $n$-qubit Pauli word, this leads to two cases, depending on whether $P_i$ commutes with $C_i s_{i-1} C_i^\dagger$:
\begin{enumerate}
    \item If $P_i$ commutes with $C_i s_{i-1} C_i^\dagger$, from Eq.~\eqref{eq:evolution}, we have
\[
\Tr{\mathcal{E}_i(s_i)\mathcal{U}_i(s_{i-1})} = \Tr{\mathcal{E}_i(s_i) C_i s_{i-1} C_i^\dagger} = \Tr{s_i\, \mathcal{E}_i(C_i s_{i-1} C_i^\dagger)}.
\]
Since $P_i$ commutes with $C_i s_{i-1} C_i^\dagger$ and $\mathcal{E}_i$ only contains Pauli operator $P_i$, we have
\[
\mathcal{E}_i(C_i s_{i-1} C_i^\dagger) = C_i s_{i-1} C_i^\dagger.
\]
Therefore, the noise does not affect $\Tr{\mathcal{E}_i(s_i)\mathcal{U}_i(s_{i-1})}$, and this term remains constant(contains no trigonometric polynomial terms).
\item If $P_i$ anticommutes with $C_i s_{i-1} C_i^\dagger$, from Eq.~\eqref{eq:evolution}, we have
\begin{equation}\label{eq:app:01}
\Tr{\mathcal{E}_i(s_i)\mathcal{U}_i(s_{i-1})} = \Tr{\mathcal{E}_i(s_i) \left(\cos(\theta_i)P^\prime_i + i\sin(\theta_i) P_i P^\prime_i\right)} = \Tr{s_i \left(\cos(\theta_i)\mathcal{E}_i(P^\prime_i) + i\sin(\theta_i) \mathcal{E}_i(P_i P^\prime_i)\right)}.
\end{equation}
where $P^\prime_i = C_i s_{i-1} C_i^\dagger$. Since $P_i$ anticommutes with $P^\prime_i$ and $P_i P_i^\prime$, $\mathcal{E}_i(P^\prime_i) = (1-2\gamma) P_i^\prime$ and $\mathcal{E}_i(P_i P^\prime_i) = (1-2\gamma) P_i P_i^\prime$
therefore $\left(\cos(\theta_i)\mathcal{E}_i(P^\prime_i) + i\sin(\theta_i) \mathcal{E}_i(P_i P^\prime_i)\right) = (1-2\gamma) \left(\cos(\theta_i)P^\prime_i + i\sin(\theta_i) P_i P^\prime_i\right)$. Therefore
$$
\Tr{\mathcal{E}_i(s_i)\mathcal{U}_i(s_{i-1})} = (1-2\gamma) \Tr{s_i\mathcal{U}_i(s_{i-1})}.
$$
Combining with Eq.~\eqref{eq:app:01}, we have
$$
\Tr{\mathcal{E}_i(s_i)\mathcal{U}_i(s_{i-1})} = \Tr{s_i \left(\cos(\theta_i)\mathcal{E}_i(P^\prime_i) + i\sin(\theta_i) \mathcal{E}_i(P_i P^\prime_i)\right)} = (1-2\gamma) \left( 
\cos(\theta_i) \Tr{s_i P^\prime_i} + \sin(\theta_i)\Tr{s_i  i P_i P^\prime_i}
\right),
$$
both $\Tr{s_i P^\prime_i}$ and $\Tr{s_i  i P_i P^\prime_i}$ are non-zero constants, therefore $\Tr{\mathcal{E}_i(s_i)\mathcal{U}_i(s_{i-1})}$ is a trigonometric polynomial of degree $1$ in $\theta_i$, and compare to noiseless case, a factor of $(1-2\gamma)$ is introduced.
\end{enumerate}

For each $i = 1,2,\cdots L$, we have shown that $\Tr{\mathcal{E}_i(s_i)\mathcal{U}_i(s_{i-1})}$ is either a constant as same as the noiseless case, or a trigonometric polynomial of degree $1$ in $\theta_i$, which introduces a factor of $(1-2\gamma)$ compared to the noiseless case $\Tr{s_i\mathcal{U}_i(s_{i-1})}$.
Therefore, we have shown that for each Pauli path $s$ with non-zero coefficient $\tilde f(s,\bm{\theta})$, the Pauli path coefficient $\tilde f(s,\bm{\theta}) = (1-2\gamma)^{\#\mathrm{trig}(s)}f(s,\bm{\theta})$, where $\#\mathrm{trig}(s)$ denotes the degree of $f(s,\bm{\theta})$ as a trigonometric polynomial. 

Therefore, the expectation value of the noisy quantum circuit can be expressed as
\begin{equation}\label{eq:noisy_expectation_value_final}
    \begin{aligned}
        \tr{\rho \tilde{O}_{\text{noisy}}}  &=\sum_{s\in \bm{P}^{L+1}_n} \tilde f(s,\bm{\theta}) \tr{s_0 \rho} = \sum_{s\in \bm{P}^{L+1}_n} (1-2\gamma)^{\#\mathrm{trig}(s)} f(s,\bm{\theta}) \tr{s_0 \rho}, \\
    \end{aligned}
\end{equation}
which illustrates a strong correlation between the degree of noise suppression and the frequency components of the circuit expectations.

Moreover, this result does not depend on the depth and the qubit number of the parameterized quantum circuits, which means that noise has a destructive impact on high-frequency components, even in shallow circuits. This further highlights the fragility of high-frequency components, as the suppression occurs even under such limited and relatively mild noise conditions.

This result directly implies that, when viewing the expectation value of a parameterized quantum circuits as a function of its parameters, the introduction of noise into the quantum circuit suppresses the high-frequency components. In other words, the higher the frequency, the more strongly it is attenuated by the noise, indicating that noise acts as a suppressor of high-frequency terms.
Based on this observation and using techniques from Ref.~\cite{aharonov2023polynomial,PhysRevLett.133.120603,rudolph2025pauli}, one can efficiently simulate noisy quantum circuits introduced noise, which is defined in Eq.~\eqref{eq:noise_channel} by computing only the low-frequency components. Note that such noise is relatively weak and sparse in the circuit, yet the entire circuit remains classically simulable efficiently.

In the next subsection, we provide a rigorous proof of the classical efficient simulability under this noise model.

\subsection{Theoretical Foundations of Low-Frequency Simulation for Noisy Quantum Circuits}\label{sec:theoretical_foundations_low_frequency_simulation_noisy_quantum_circuits}

In this section, we will provide a rigorous proof of the classical efficient simulability of noisy quantum circuits under the noise model described in Eq.~\eqref{eq:noise_channel}. Specifically, we will show that noisy quantum circuits can be well approximated by low-frequency terms, and the number of terms required is polynomial in nature, allowing for classical efficient simulation.

We suppose the observable $O$ is a normalized Pauli operator.
Using Eq.~\eqref{eq:noisy_expectation_value_final}, the expectation value of the noisy circuit can be expressed as:
\begin{equation}
    \begin{aligned}
        \langle O_{\text{noisy}}\rangle &=\tr{\rho \tilde{O}_{\text{noisy}}} \\
          &=\sum_{s\in \bm{P}^{L+1}_n} \tilde f(s,\bm{\theta}) \tr{s_0 \rho} \\
        & = \sum_{s\in \bm{P}^{L+1}_n} (1-2\gamma)^{\#\mathrm{trig}(s)} f(s,\bm{\theta}) \tr{s_0 \rho}. \\
    \end{aligned}
\end{equation}
We could truncate the summation to only include low-frequency components, i.e., $\#\mathrm{trig}(s) \leq \eta$, where $\eta$ is a positive integer. 
Then we have
\begin{equation}\label{eq:noisy_expectation_value_20}
    \begin{aligned}
        \langle O_{\text{noisy}}\rangle^{(\eta)} &=\sum_{s\in \bm{P}^{L+1}_n(\eta)} (1-2\gamma)^{\#\mathrm{trig}(s)} f(s,\bm{\theta}) \tr{s_0 \rho} \\
    \end{aligned}
\end{equation}
where $\bm{P}^{L+1}_n(\eta)$ is the set of Pauli paths s.t. $\forall s \in \bm{P}^{L+1}_n(\eta)$ satisfies $\#\text{trig}(s) \leq \eta$.  

Next, we compute the mean error between $\langle O_{\text{noisy}}\rangle$ and $\langle O_{\text{noisy}}\rangle^{(\eta)}$ with respect to parameter $\bm{\theta}$, which is defined as
\begin{equation}
    \begin{aligned}\label{eq:mean_error22}
        \mathbb{E}_{\bm\theta} \left| \langle O_{\text{noisy}}\rangle - \langle O_{\text{noisy}}\rangle^{(\eta)}\right|^2 & = \mathbb{E}_{\bm\theta} \left| \sum_{s\in \bm{P}^{L+1}_n\setminus \bm{P}^{L+1}_n(\eta)} (1-2\gamma)^{\#\mathrm{trig}(s)} f(s,\bm{\theta}) \tr{s_0 \rho} \right|^2\\
        & =  \mathbb{E}_{\bm\theta} \sum_{s\in \bm{P}^{L+1}_n\setminus \bm{P}^{L+1}_n(\eta)} \left| (1-2\gamma)^{\#\mathrm{trig}(s)} f(s,\bm{\theta}) \tr{s_0 \rho} \right|^2\\
        & \leq (1-2\gamma)^{2\eta} \mathbb{E}_{\bm\theta}  \sum_{s\in \bm{P}^{L+1}_n\setminus \bm{P}^{L+1}_n(\eta)} \left|  f(s,\bm{\theta}) \tr{s_0 \rho} \right|^2 \\
        & \leq (1-2\gamma)^{2\eta} \mathbb{E}_{\bm\theta}  \sum_{s\in \bm{P}^{L+1}_n} \left|  f(s,\bm{\theta}) \tr{s_0 \rho} \right|^2 \\
        & = (1-2\gamma)^{2\eta} \mathbb{E}_{\bm\theta} \left| \sum_{s\in \bm{P}^{L+1}_n}  f(s,\bm{\theta}) \tr{s_0 \rho} \right|^2\\
        & \leq (1-2\gamma)^{2\eta} \max_{\bm\theta} \left| \sum_{s\in \bm{P}^{L+1}_n}  f(s,\bm{\theta}) \tr{s_0 \rho} \right|^2\\
        & \leq (1-2\gamma)^{2\eta}\left\|  O \right\|^{2}_2,\\
    \end{aligned}
\end{equation}
where $\left\|  O \right\|_2$ is the operator norm~(i.e., its largest singular value).
The equality in the second line and the fifth line holds because for $s \neq s^\prime$, $f(s,\bm{\theta})$ and $f(s^\prime,\bm{\theta})$ represent different trigonometric polynomials, which are orthogonal to each other, i.e.
$$
\mathbb{E}_{\bm\theta}  f(s,\bm{\theta}) f(s^\prime,\bm{\theta}) = 0.
$$
Since the observable \( O \) is a Pauli operator, its operator norm is bounded by a constant, i.e., \( \|O\|_2 \leq C \) for some constant \( C \).  
If \( \eta = \mathcal{O}\left(\ln(\varepsilon)/\ln(1-2\gamma)\right) \), then we have
\begin{equation}
    \begin{aligned}
        \mathbb{E}_{\bm\theta} \left| \langle O_{\text{noisy}}\rangle - \langle O_{\text{noisy}}\rangle^{(\eta)}\right|^2 \leq \varepsilon.
    \end{aligned}
\end{equation}

Next we will show that, if $\eta = \mathcal{O}\left(\ln(\varepsilon)/\ln(1-2\gamma)\right)$,
the number of non-zero Pauli paths in $\bm{P}^{L+1}_n(\eta)$ is $\mathcal{O}\left(2^\eta\right)$ which is polynomial if we suppose the noise rate $\gamma$ is constant.
We explicitly construct a procedure to enumerate all non-zero Pauli paths \( s \) satisfying \( \#\mathrm{trig}(s) \leq \eta \). Through this enumeration process, we can conclude that the number of such Pauli paths is \( \mathcal{O}\left(2^\eta\right) \).

For each path $s = (s_0,\cdots, s_L)$, we have $s_L = O$, otherwise the coefficient of this path is zero.
As for $U_L(\theta_L)^\dagger s_L U_L(\theta_L)$, if $U_L(\theta_L)^\dagger s_L U_L(\theta_L)$ is a Pauli operator, then $s_{L-1} = U_L(\theta_L)^\dagger s_L U_L(\theta_L)$, otherwise the coefficient of this path is zero. In another case, $U_L(\theta_L)^\dagger s_L U_L(\theta_L)$ is the sum of two Pauli operators, with coefficients $\sin(\theta_L)$ and $\cos(\theta_L)$, then $s_{L-1}$ has two choices. When we use $s_i$ to construct $s_{i-1}$, if the path splits (which means $s_{i-1}$ has two choices), then a sin or cos term is introduced; otherwise, no sin or cos term is introduced. Since the number of sin or cos terms is at most $\eta$, therefore when we construct the path $s$ from $s_L$ to $s_0$ (following the process above), it could be split at most $\eta$ times, which means the number of non-zero Pauli paths is at most $2^\eta$, and to construct all non-zero Pauli paths, we need to enumerate all possible split points, which is at most $\eta$, therefore, the complexity of the enumeration is \( 2^\eta \).
Note that for each Pauli path \(s\), the computational complexity of evaluating its coefficient \(f(s, \bm{\theta})\) is polynomial in the number of qubits \(n\) and the circuit depth \(L\).  
Therefore, the cost of computing Eq.~\eqref{eq:noisy_expectation_value_20} is \(\mathcal{O}(2^{\eta}) \cdot \mathrm{poly}(n, L)\),  
which implies an overall runtime that is polynomial in the circuit depth \(L\), the number of qubits \(n\), and \(1/\varepsilon\).

If a noisy circuit not only contains noise defined in Eq.~\eqref{eq:noise_channel}, but also introduces additional Pauli noise elsewhere, then the coefficient of the Pauli path $s$ in the noise case $\tilde f(s,\bm\theta)$ has the following relationship with the noiseless case $f(s,\bm\theta)$:
\begin{equation}\label{eq:mean_error220000}
    \left|\tilde f(s,\bm\theta)\right| \leq  \left|(1-2\gamma)^{\#\mathrm{trig}(s)} f(s,\bm\theta)\right|.
\end{equation}
This is because the additional Pauli noise introduces an extra suppression factor for the coefficient of each Pauli path \( s \), similar to how noise in Eq.~\eqref{eq:noise_channel} contributes a suppression factor of \( (1 - \gamma) \).
Therefore, the mean error between the expectation value of the noisy quantum circuit and its low-frequency approximation is bounded by
\begin{equation}
    \begin{aligned}\label{eq:mean_error22333}
        \mathbb{E}_{\bm\theta} \left| \langle O_{\text{noisy}}\rangle - \langle O_{\text{noisy}}\rangle^{(\eta)}\right|^2 & \leq  \mathbb{E}_{\bm\theta} \sum_{s\in \bm{P}^{L+1}_n\setminus \bm{P}^{L+1}_n(\eta)} \left| (1-2\gamma)^{\#\mathrm{trig}(s)} f(s,\bm{\theta}) \tr{s_0 \rho} \right|^2\\
        & \leq  (1-2\gamma)^{2\eta}\left\|  O \right\|^{2}_2,\\
    \end{aligned}
\end{equation}
which is similar to Eq.~\eqref{eq:mean_error22}. 
Therefore, as long as noise defined in Eq.~\eqref{eq:noise_channel} is introduced after each rotation gate, the above simulation method remains efficient under the condition of small average error with respect to $\bm{\theta}$. The results of this section can be summarized in the following theorem.

\begin{theorem}
    For a parameterized quantum circuit whose unitary operator can be written as $C(\bm\theta) = U_L(\theta_L) \cdots U_1(\theta_1)$, where each $U_i(\theta_i) = \exp{-i P_i\theta_i/2}$ is a rotation gate with $P_i \in \{I,X,Y,Z\}^{\otimes n}$ being an $n$-qubit Pauli operator and $\bm{\theta} = (\theta_1,\cdots,\theta_L) \in \mathbb{T}^L$ is the vector of rotation angles. The circuit is followed by a measurement of an observable $O$ which $\|O\|_\infty$ is bounded by a constant. 
Noise introduced after each rotation gate $\exp{-i P_i\theta_i/2}$ is modeled by a Pauli channel $\mathcal{E}_i$ defined in Eq.~\eqref{eq:noise_channel} with noise rate $\gamma$. 
    The noisy quantum circuit expectation value $\langle O_{\text{noisy}}\rangle$ can be approximated by low-frequency terms $\langle O_{\text{noisy}}\rangle^{(\eta)}$ defined in Eq.~\eqref{eq:noisy_expectation_value_20}. Specifically, if $\eta = \mathcal{O}\left(\ln(\varepsilon)/\ln(1-2\gamma)\right)$, then the mean error between the expectation value of the noisy quantum circuit and its low-frequency approximation is bounded by \( \varepsilon \), i.e. $\mathbb{E}_{\bm\theta} \left| \langle O_{\text{noisy}}\rangle - \langle O_{\text{noisy}}\rangle^{(\eta)}\right|^2 \leq \varepsilon$.
    The computational complexity of the low-frequency approximation $\langle O_{\text{noisy}}\rangle^{(\eta)}$ is $\mathcal{O}(2^{\eta}) \cdot \mathrm{poly}(n, L)$, which means that it can be simulated classically efficiently.
\end{theorem}

\section{Additional Numerical Results Related to Section~\ref{sec:QNN_Robustness}}\label{appendix:application1_numerical}

In Section~\ref{sec:QNN_Robustness}, we have shown the Fourier spectrum of the QNN function $f_{\bm{\theta}}(x)$ across training iterations under different noise models.
To provide a more comprehensive view of the evolution of Fourier amplitudes across different frequencies in Fig.~\ref{fig:pic_QNN_Robustness22_a} and Fig.~\ref{fig:pic_QNN_Robustness22_b}, we present in Fig.~\ref{fig:whole_picture_MPO_QEMappendix} the results for all frequencies from 1~Hz to 21~Hz, illustrating the amplitude dynamics under various types of noise.

Furthermore, in Section~\ref{sec:QNN_Robustness}, we only consider the noise introduced in the rotation gate, which encodes the input $x$. However, in practice, noise is introduced in every gate. Therefore, for every gate except the rotation gate we add noise before, we add depolarizing noise with rate $0.001$ for the single qubit gate and depolarizing noise with rate $0.01$ for double qubit gate, and present the Fourier spectrum of the QNN function $f_{\bm{\theta}}(x)$ across training iterations in Fig.~\ref{fig:whole_picture_MPO_QEMappendix2}.

\begin{figure*}[htb!]
 \centering
 \includegraphics[width = \columnwidth]{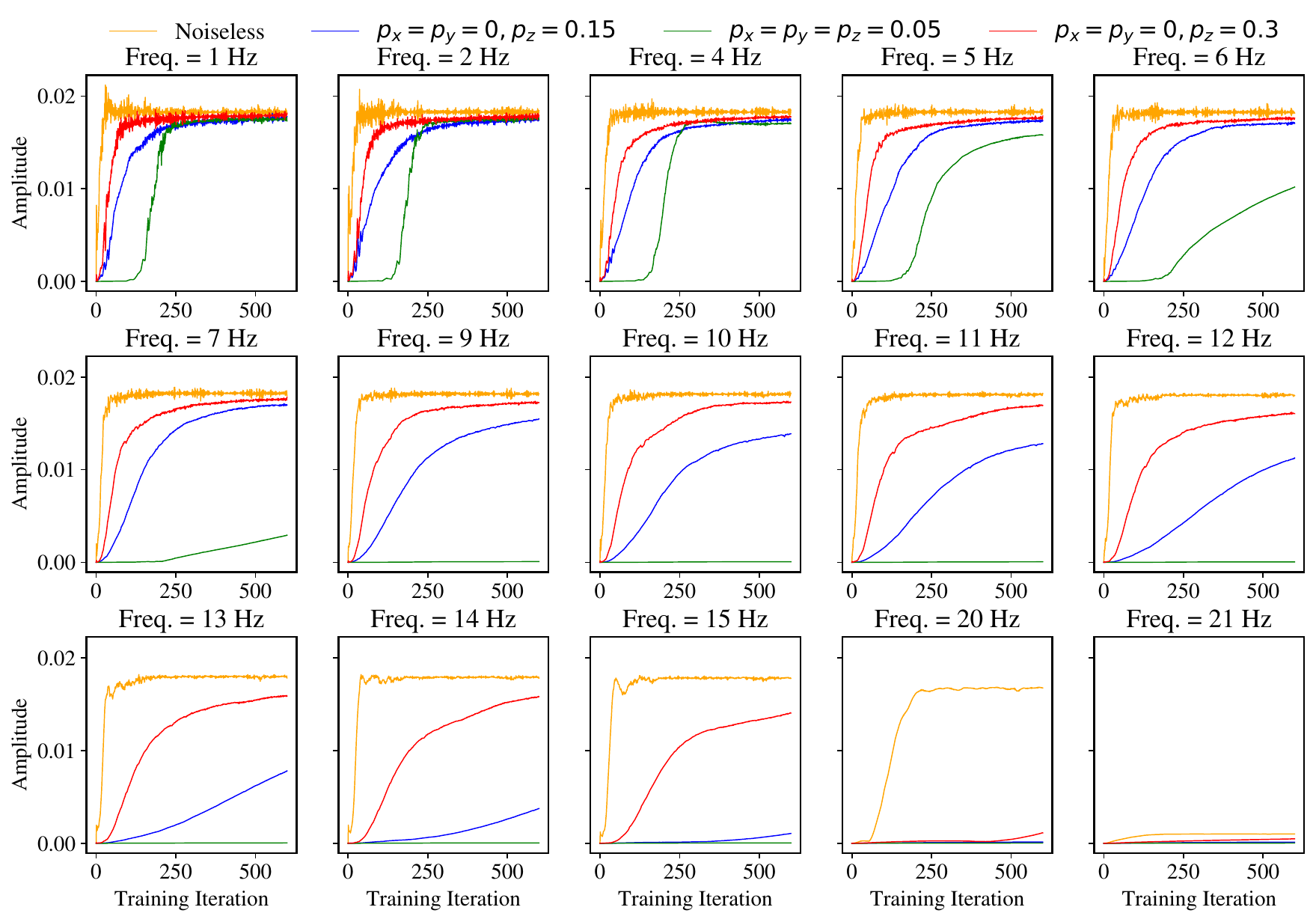}
 \caption{\justifying
 Evolution of Fourier amplitudes at each frequency (from 1Hz to 21 Hz) during QNN training under various noise models, revealing how noise strength and channel type jointly govern the convergence rate and steady-state amplitude across training iterations. }
 \label{fig:whole_picture_MPO_QEMappendix}
\end{figure*}

\begin{figure*}[htb!]
 \centering
 \includegraphics[width = \columnwidth]{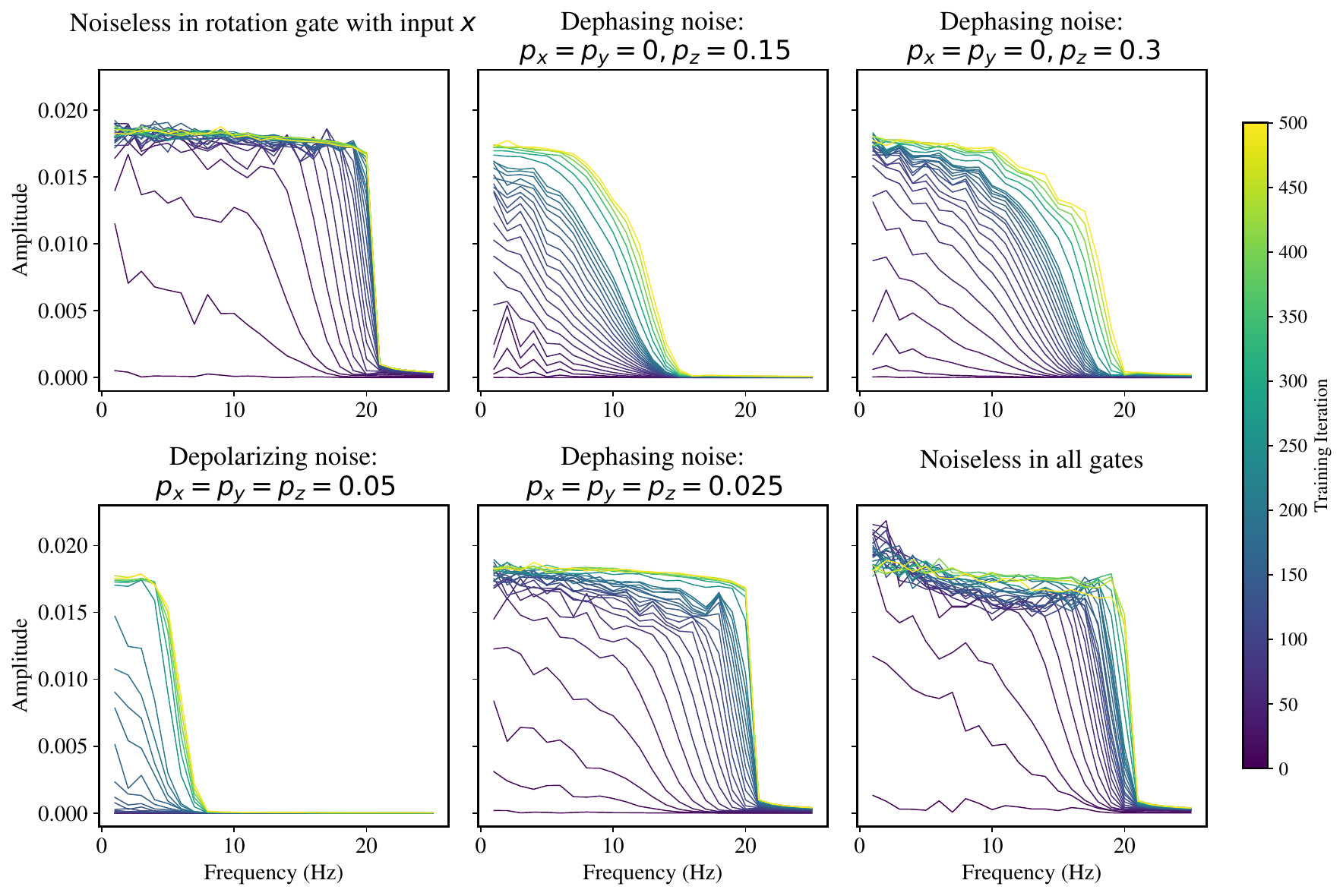}
 \caption{\justifying
    Fourier spectrum of the QNN function $f_{\bm{\theta}}(x)$ across training iterations under different noise models of the encoding rotation gate of input $x$ (described in the title of each subplot). And for other gates except the encoding rotation gate, we add depolarizing noise with rate $0.001$ for the single qubit gate and depolarizing noise with rate $0.01$ for the double qubit gate for all the subfigures except the last one. For the last subfigure, each gate in the QNN is noiseless.
        }
 \label{fig:whole_picture_MPO_QEMappendix2}
\end{figure*}

\end{document}